\documentclass[journal,letterpaper,onecolumn,11pt]{IEEEtran}
\usepackage{hyperref}
\usepackage{tikz}
\usetikzlibrary{positioning, arrows.meta, fit, backgrounds, calc}
\usetikzlibrary{arrows.meta}
\usepackage{amsmath}
\usepackage{amsthm}
\newtheorem{assumption}{Assumption}
\hypersetup{
 colorlinks,
 linkcolor={blue!100!black},
 citecolor={blue!100!black},
 urlcolor={blue!80!black}
}
\usepackage{accents}
\newlength{\dhatheight}

\def\BibTeX{{\rm B\kern-.05em{\sc i\kern-.025em b}\kern-.08em
    T\kern-.1667em\lower.7ex\hbox{E}\kern-.125emX}}
\usepackage{stmaryrd} 

\newcommand{\bE}{\mathbb{E}}

\usepackage{booktabs}

\newcommand{\cL}{\mathcal{L}}

\newcommand{\cT}{\mathcal{T}}

\newcommand{\cV}{\mathcal{V}}

\newcommand{\cY}{\mathcal{Y}}

\newcommand{\bolda}{\mathbf{a}}
\newcommand{\boldb}{\mathbf{b}}

\newcommand{\boldt}{\mathbf{t}}
\newcommand{\boldu}{\mathbf{u}}

\newcommand{\boldx}{\mathbf{x}}
\newcommand{\boldy}{\mathbf{y}}
\newcommand{\boldz}{\mathbf{z}}

\newcommand{\boldA}{\mathbf{A}}
\newcommand{\boldB}{\mathbf{B}}

\newcommand{\boldD}{\mathbf{D}}

\newcommand{\boldT}{\mathbf{T}}

\newcommand{\boldX}{\mathbf{X}}
\newcommand{\boldY}{\mathbf{Y}}
\newcommand{\boldZ}{\mathbf{Z}}

\usepackage{tikz}
\newif\ifFULL
\FULLfalse
\usepackage[utf8]{inputenc} 
\usepackage[T1]{fontenc}
\usepackage{url}
\usepackage{bbm}
\usepackage{ifthen}
\usepackage[cmex10]{mathtools}

\usepackage{verbatim}
\usepackage{amsthm,amssymb}
\usepackage{xcolor}
\usepackage{graphicx}
\usepackage{amssymb}
\usepackage{epstopdf}

\newtheorem{theorem}{Theorem}

\newtheorem{remark}{Remark}

\newtheorem{proposition}{Proposition}

\newtheorem{lemma}{Lemma}
\newtheorem{definition}{Definition}
\newtheorem{construction}{Construction}
\newtheorem{corollary}{Corollary}
\usepackage[style=ieee,citestyle=numeric-comp,sorting=nty]{biblatex}
\usepackage{algorithm}
\usepackage{algorithmic}
\newcommand{\red}[1]{\textcolor{red}{#1}}

\begin{document}
\title{The Capacity of the Noisy Torn Paper Channel and the Noisy Shuffling Channel}


\author{%
  \IEEEauthorblockN{\textbf{Junsheng Liu} and \textbf{Netanel Raviv}}  \IEEEauthorblockA{\\~\\Department of Computer Science and Engineering\\
                    Washington University in St Louis, St Louis, MO, USA\\
                    \texttt{junsheng,netanel.raviv@wustl.edu}}
}


\maketitle

\newcommand{\ind}{\mathbf{1}}
\newcommand{\Csh}{\mathsf{C}_p}
\newcommand{\Var}{\operatorname{Var}}
\newcommand{\wt}{\operatorname{wt}}

\begin{abstract}
The torn paper channel breaks a codeword into 
fragments and returns all fragments without their original
order. 
The noisy torn paper channel additionally
introduces bit substitution errors to each fragment.
Recovering the message therefore requires accounting for both
bit errors and uncertainty about the fragment order.
We establish a capacity upper bound for the noisy torn paper
channel that matches the achievable rate of Ravi, Vahid, and
Shomorony and resolves their conjecture.
Our results also determine the capacity of the noisy shuffling channel of Shomorony and Heckel, in which all fragments have the same length, and resolve their open question.
\end{abstract}
\thispagestyle{empty}

\section{Introduction}
\label{sec:introduction}
Consider an information sequence encoded on a long strip
of paper.
The strip is torn at random locations, and the
resulting fragments are shuffled before reaching the
receiver. 
Every bit remains intact, and the order of
bits within each fragment is preserved, but the fragments'
original positions are unknown. 
Recovering the message
therefore requires resolving uncertainty about how the
fragments relate to the original sequence. 
This synchronization difficulty
is central to \emph{torn paper coding}. 
Motivated by DNA storage~\cite{sima2023error}, Shomorony and Vahid introduced the torn paper
channel (TPC) and characterized its capacity in~\cite{shomorony2021torn}.
The torn paper channel is also related to other media susceptible to physical
breaks~\cite{wang2025secure}.

The main question following the capacity results in~\cite{shomorony2021torn}
is how to construct capacity achieving codes with efficient
encoding and decoding. 
The interleaved-pilot construction
in~\cite{shomorony2021torn} provides a structured scheme
with provable reliability, but leaves a gap between its
achievable rate and capacity. 
Other approaches to the torn paper channel include nested Varshamov--Tenengolts
codes~\cite{nassirpour2023dna} and nested hash
codes~\cite{jiao2025efficient}.
However, these works do not establish that the decoding
error probability tends to zero as the codeword length
goes to infinity.
Liu and Raviv improved the achievable rates by identifying
pilot positions locally within fragments, allowing shorter
fragments to be used for decoding~\cite{liu2026improved,Junsheng2026torn}.
Their subsequent construction attains capacity through
successive local alignment and pilot recycling: recovered
codeword bits provide references for aligning progressively
shorter fragments~\cite{liu2026capacity}. 
In parallel, another similar code construction appeared in~\cite{wang2025secure}.

DNA and other storage systems may suffer both physical fragmentation and errors in the recovered
fragments content, such as bit flips, 
motivating the noisy torn paper channel.
Ravi, Vahid, and Shomorony~\cite{ravi2024recovering}
studied two extensions of the torn paper channel.
In the first extension, called the torn paper channel with lost pieces, fragments may be lost with probabilities
depending on their lengths, while the bits in the
remaining fragments are unchanged. 
They characterized
the capacity of this channel and showed that random
coding achieves it.
In the second extension, called the \textit{noisy} torn paper channel, all fragments are retained, but their
bits undergo substitution errors through a binary
symmetric channel (BSC). 
For this noisy TPC, they used random
coding to establish an achievable rate and derived
a capacity upper bound. 
They conjectured that the achievable
rate is the capacity, leaving a matching converse
to be established.
Walter, Abu-Sini, Weinhardt, and Wachter-Zeh in~\cite{walter2026coding} developed
coding schemes for the noisy TPC that use fixed markers
or markers computed from the encoded data to help
reassemble fragments in the presence of bit errors,
and evaluated their performance for finite codeword
lengths.
These schemes provide approaches to decoding noisy
fragments, but their performance is evaluated through
simulations without a proof that the decoding error
probability tends to zero as the codeword length~$n$ tends to infinity. 
The gap between the known 
rate achievable by random encoding and capacity upper bound also remains unresolved
in general.

A closely related model is the noisy shuffling channel (NSC) of Shomorony and
Heckel~\cite{shomorony2019capacity}.
In this channel, a codeword is divided into fragments of equal length; each fragment passes through a
BSC, and the decoder observes the noisy fragments without their order.
For this channel, a gap remained between the known achievable rate and the capacity upper bound under some parameter regime~\cite{shomorony2019capacity}.

Other related fragmentation models impose different constraints on
the breaks and on the fragments available for recovery.
Motivated by forensic identification of 3D printed objects,
break-resilient codes guarantee recovery from the collection
of fragments produced by a prescribed
number of breaks~\cite{wang2024breakresilientcodesforensic3d,wang2026break,wang2025secure}.
Single-fragment forensic codes instead require recovery from
only one sufficiently large fragment, without access to the
remaining pieces~\cite{liu2025single,liu2025single2,sun2026singlefragmentforensiccodingmultidimensional}.
Adversarial torn-paper codes allow the break locations to be
chosen adversarially, subject to constraints on fragment
lengths~\cite{barlev}.
The related sliced channel~\cite{sima2023error} considers unordered collections of fixed-length strings
that may also undergo errors, providing a related setting for
studying the combined effects of errors within strings and
the loss of their order.

In this paper, we study the noisy TPC, where a codeword
undergoes bit substitution errors and is then divided into
fragments whose lengths are sampled from a prescribed
distribution; note that the original TPC as defined by~\cite{shomorony2021torn} specifies this distribution as geometric. 
The decoder receives all fragments with
the order of bits within each fragment preserved, but
without knowing their original positions in the codeword.
Our goal is to determine the capacity by accounting for
both substitution errors and the loss of fragment order.
The main difficulty is that these two sources of uncertainty
interact: bit errors can make fragments harder to place,
while uncertainty about their positions can make the bit
errors harder to resolve.
We establish a converse that addresses both difficulties
and matches the achievable rate of Ravi, Vahid, and
Shomorony~\cite{ravi2024recovering}, thereby resolving their
capacity conjecture for the noisy TPC.
The noisy shuffling channel~\cite{shomorony2019capacity} is the special case of the noisy TPC in which all fragments have the same length, so our results also
close the gap between the known achievable rate and capacity upper bound
for that channel.

Section~\ref{sec:model} defines the noisy TPC and the NSC, which is a
special case of the noisy TPC, states the assumptions on the fragment lengths, and presents our main results: the capacity of the noisy TPC, which resolves the conjecture of~\cite{ravi2024recovering}, and the
capacity of the NSC, which answers the open question
in~\cite{shomorony2019capacity}.
Section~\ref{sec:capacity-proof} proves the upper bound on the capacity of the noisy TPC by assuming a bound on the information carried by an unordered collection of noisy fragments of equal length.
Section~\ref{sec:sparse} proves this bound and uses it to derive the capacity of the NSC.
Finally, Section~\ref{sec:scope} discusses directions for future work.

\section{Model and Main Results}
\label{sec:model}
\label{sec:main}
In this section, we first describe how the BSC introduces substitution
errors into a codeword and how independently sampled
fragment lengths determine the cut locations.
We then specify the unordered multiset available to
the decoder, which must recover the message without
knowing the original positions of the noisy fragments.
To formulate this recovery requirement, we introduce
the code rate and average decoding error probability
and define capacity through rates achievable with
vanishing error as the codeword length $n\to\infty$.
Finally, we state the assumptions on the sampled
fragment lengths, present the capacity upper bound
for the noisy TPC,  provide a corollary for
independent-cut noisy TPC, and define the NSC as a special case of the noisy TPC and determine its capacity.

\subsection{The noisy torn paper channel and the noisy shuffling channel}\label{section:noisy tpc}

We consider the noisy TPC
of~\cite[Section~III]{ravi2024recovering}. A length-$n$
codeword first passes through a binary symmetric channel (BSC), and the resulting noisy
sequence is then divided into fragments.
The fragment lengths are generated sequentially from a
prescribed distribution. 
Let the codeword be $\boldx\in\{0,1\}^n$ and let
$N_{n,1},N_{n,2},\ldots$ be independent and identically distributed  positive
integer-valued random variables, whose distribution does not depend on~$\boldx$.
Their distribution may depend on $n$, but $N_{n,1},N_{n,2},\ldots$ follow the same
distribution for each fixed \(n\).
Each $N_{n,j}$ specifies the sampled length of the $j$th fragment.
Starting from the beginning of the codeword, we use the first
sampled length $N_{n,1}$ to determine where the first fragment
ends.
If bits remain after the first fragment, we use the next
independently sampled length $N_{n,2}$ to determine the next
fragment boundary and continue in the same way.
If a sampled length is at least the number of
remaining bits, all remaining bits form the final fragment and
the procedure stops. 
We do not reject or resample a length that
extends beyond the end of the codeword; instead, we truncate that
final fragment to contain exactly the remaining bits. 
In addition to fragmentation, the channel introduces independent
bit-flip errors through a BSC.
Throughout the paper, let $p\in(0,1/2)$ be fixed.
The BSC independently flips each bit with probability $p$ and leaves it unchanged with probability $1-p$.
After passing through the fragmentation process, the noisy TPC returns an unordered multiset $\cY$ of all fragments. 
The order of bits within each fragment is
preserved, and no fragment is lost. 
The decoder's task is to recover the original message from
this unordered collection of noisy fragments, despite the
unknown fragment positions and the substitution errors.
We evaluate codes by their rate and average decoding
error probability, defined as follows.

A length-$n$ noisy torn paper code consists of $m_n$ messages, a deterministic encoder
$\operatorname{Enc}:\{1,\ldots,m_n\}\to\{0,1\}^n$, and a decoder $\operatorname{Dec}$
that maps each possible multiset of noisy fragments to a message estimate. The rate of the code is $R_n=(\log m_n)/n$, where all logarithms are base two.
For $M\sim\textrm{Unif}(\{1,\ldots,m_n\})$,
the average decoding error probability is
\begin{align*}
P_{e,n}
&=\Pr\bigl(\operatorname{Dec}(\cY)\neq M\bigr)
=\frac{1}{m_n}\sum_{a=1}^{m_n}
\Pr\bigl(\operatorname{Dec}(\cY)\neq a\mid M=a\bigr).
\end{align*}
A rate $R\geq0$ is achievable if there exists a deterministic
sequence of codes satisfying $\liminf_{n\to\infty}R_n\geq R$
and $P_{e,n}\to0$ as $n\to\infty$.
The capacity $C$ is the supremum of achievable rates.
Our goal is to characterize this capacity for the noisy TPC.
The following definition specifies its fragment boundaries and
the truncation of the length for final sampled fragment.
\begin{definition}
\label{def:renewal-model}
 Let
$\{N_{n,j}\}_{j\geq1}$ be independent, identically distributed positive
integer-valued random variables, independent of the codeword
and the BSC noise.
For every integer $k\geq 0$, let $S_{n,k}$ denote the cumulative
length of the first $k$ $N_{n,j}$'s with $S_{n,0}=0$, and let $K_n$
be the first index at which the cumulative length reaches or
exceeds~$n$, i.e., 
\begin{align*}
 S_{n,k}=\sum_{j=1}^{k}N_{n,j}
\mbox{ and }
 K_n=\min\{k\geq 1:S_{n,k}\geq n\}.
\end{align*}
  The codeword is divided into $K_n$ fragments,
where $1\leq K_n\leq n$, and the length $L_{n,j}$ of the
$j$th fragment, for $j\in[K_n]\triangleq\{1,\ldots,K_n\}$, is
\begin{align*}
 L_{n,j}&=
 \begin{cases}
 N_{n,j}, & 1\leq j<K_n,\\
 n-S_{n,K_n-1}, & j=K_n.
 \end{cases}
\end{align*}
In particular, the last fragment is truncated if
$S_{n,K_n}>n$, and $\sum_{j=1}^{K_n}L_{n,j}=n$.
The breaks occur after bits $S_{n,j}$ for
$j\in[K_n-1]$.
Write $\boldT=(T_1,\ldots,T_{n-1})$ for the \emph{cut pattern},
where $T_i=1$ when a break occurs after bit $i$ and $T_i=0$
otherwise.
For each $j\in[K_n]$, the $j$th fragment occupies positions
$S_{n,j-1}+1,\ldots,\min\{S_{n,j},n\}$.
\end{definition}
An important special case of the noisy TPC is the noisy shuffling
channel (NSC), in
which all fragments have the same deterministic length.
\begin{definition}
\label{def:nsc}
The \emph{noisy shuffling channel} (NSC) is the noisy TPC in which
every sampled length equals the same deterministic positive
integer~$\ell_n$, that is, $N_{n,j}=\ell_n$ for every~$j$, where
$\ell_n$ divides~$n$. The codeword is then divided into
$K_n=n/\ell_n$ consecutive fragments of length~$\ell_n$, no fragment
is truncated, and the decoder observes the unordered multiset $\cY$ of
the noisy fragments. Codes, rates, and capacity for the NSC are
defined as for the noisy TPC, with codeword lengths~$n$ divisible
by~$\ell_n$.
\end{definition}

\begin{remark}
Other than the NSC (Definition~\ref{def:nsc}), the noisy TPC includes several additional special cases obtained by
specifying the distribution from which each fragment length is sampled. Geometric sampled lengths provide the fragmentation
rule of the torn paper channel
in~\cite{shomorony2021torn} (termed \emph{independent-cut} torn paper channel herein), where each boundary between
consecutive bits is cut independently with probability
$q_n$. 
Together with the BSC noise, this yields the
independent-cut torn paper channel with substitution errors.
Another special case, that is similar to~\cite{barlev}, is obtained by sampling successive
lengths independently and uniformly from
$\{a_n,a_n+1,\ldots,b_n\}$, where
$1\leq a_n\leq b_n$ are integers.  
These cases differ in their fragment-length distributions,
while retaining the same BSC substitution errors and an unordered multiset of noisy fragments as the decoder observation. 
\end{remark}

While our results apply for a wide range of fragment length distributions, we do follow the fragment-length
assumptions from~\cite[Section~II-A]{ravi2024recovering}, as follows. 

\begin{assumption}\label{remark:assumptions}
Let $\ell_n=\bE[N_{n,1}]$ and assume that there
is a fixed constant $\alpha\in(0,\infty)$ such that
\begin{align}
\ell_n&=\bE[N_{n,1}],
&
\frac{\log n}{\ell_n}
&\longrightarrow\alpha\in(0,\infty),
&
\bE[N_{n,1}^2]&=\Theta((\log n)^2)
\qquad \text{when}\ n\to\infty.
\label{eq:renewal-assumptions}
\end{align}
\end{assumption}

Notice that Assumption~\ref{remark:assumptions} implies that \textit{all} fragment lengths~$\{N_{n,j}\}_{j\geq1}$ satisfy~\eqref{eq:renewal-assumptions} since the fragment lengths are identically distributed. 
Additionally, the requirement about the second moment of each~$N_{n,j}$ in~\eqref{eq:renewal-assumptions} is due to technical difficulties.

\subsection{Main capacity results}
For $u\in[0,1]$, define
\begin{equation}
 r_n(u)=\frac{\bE[(uN_{n,1}-\log n)^+]}{\ell_n},
 \label{eq:reward-rate}
\end{equation}
where $\ell_n=\bE[N_{n,1}]$ as defined in~\eqref{eq:renewal-assumptions} and  $(uN_{n,1}-\log n)^+=\max\{uN_{n,1}-\log n,0\}$.  
Let $\boldX\in\{0,1\}^n$ denote a random input to the noisy TPC, independent of the BSC noise and the fragment lengths.
The following theorem bounds the normalized
mutual information of the noisy TPC in terms of its sampled
fragment-length distribution.

\begin{theorem}
\label{thm:renewal-capacity}
Consider the noisy TPC with substitution probability
$p\in(0,1/2)$, and let $\kappa=1-h_2(p)$, where $h_2(p)=-p\log p-(1-p)\log(1-p)$ is the binary entropy
function.
Under Assumption~\ref{remark:assumptions}, we have that
\begin{equation}
 \limsup_{n\to\infty}\frac1n
 \sup_{P_{\boldX}}I(\boldX;\cY)
 \le\limsup_{n\to\infty}r_n(\kappa).
 \label{eq:renewal-converse}
\end{equation}
\end{theorem}
When $r_n(\kappa)$ converges as $n\to\infty$,
the converse in Theorem~\ref{thm:renewal-capacity}
matches the known achievable rate and determines
the capacity, as stated below.

\begin{corollary}
\label{cor:capacity-limit}
Under the assumptions of
Theorem~\ref{thm:renewal-capacity},
if $\lim_{n\to\infty}r_n(\kappa)$ exists, then
the capacity of the noisy TPC is $C=\lim_{n\to\infty}r_n(\kappa)$.
\end{corollary}

\begin{proof}
Fano's inequality implies that every rate achievable
with vanishing average decoding error as $n\to\infty$
is at most the limiting upper bound on mutual
information per codeword bit. 
Therefore, Theorem~\ref{thm:renewal-capacity} implies that
\begin{align*}
C
&\leq
\limsup_{n\to\infty}\frac1n
\sup_{P_{\boldX}}I(\boldX;\cY)
\overset{\eqref{eq:renewal-converse}}{\leq}
\lim_{n\to\infty}r_n(\kappa).
\end{align*}
For the converse inequality, the random-coding
argument in~\cite[Section~III-A,
eqs.~(45)--(47)]{ravi2024recovering}
establishes achievability
of every nonnegative rate strictly below
$\lim_{n\to\infty}r_n(\kappa)$.
Combining the two inequalities proves the claim.
\end{proof}
\begin{corollary}
\label{cor:geometric}
Consider the independent-cut torn paper channel with substitution errors, with
break probability~$q_n$ satisfying
$q_n\log n\to\alpha\in(0,\infty)$ and substitution probability
$p\in(0,1/2)$. Let $C(\alpha,p)$ denote its capacity.
Theorem~\ref{thm:renewal-capacity}, together with the achievable
rate in~\cite[Theorem~3]{ravi2024recovering}, implies that
\begin{equation}
 C(\alpha,p)=\kappa e^{-\alpha/\kappa}.
 \label{eq:geometric-capacity}
\end{equation}
\end{corollary}
Our results also determine the capacity of the NSC in
Definition~\ref{def:nsc}. 
Its deterministic sampled lengths satisfy
$\bE[N_{n,1}^2]=\ell_n^2$, so Assumption~\ref{remark:assumptions}
reduces to $\log n/\ell_n\to\alpha\in(0,\infty)$.
We prove the following theorem in
Section~\ref{sec:sparse}.
\begin{theorem}
\label{thm:nsc}
Consider the NSC with substitution probability $p\in(0,1/2)$. Under Assumption~\ref{remark:assumptions}, we have
\begin{equation}
 \limsup_{n\to\infty}\frac{1}{n}\sup_{P_{\boldX}}I(\boldX;\cY)\leq(1-h_2(p)-\alpha)^+.
 \label{eq:nsc-converse}
\end{equation}
Together with the achievable rate of~\cite{shomorony2019capacity},
it follows that the capacity of the NSC is $(1-h_2(p)-\alpha)^+$.
\end{theorem}

\section{Capacity of the noisy torn paper channel}
\label{sec:partition}
\label{sec:renewal}
\label{sec:capacity-proof}
The proof of Theorem~\ref{thm:renewal-capacity} relies on a bound for
a simpler setting, in which all fragments have the same length.
Proposition~\ref{prop:uniform} below bounds the information carried by
an unordered collection of such noisy fragments.
We  use the following quantity to express
the information bound for fragments of the same length.
\begin{definition}\label{definition:Csh}
For integers $m\geq1$ and $k\geq1$, let $\Csh(m,k)$ be the
supremum of $I(\boldX;\cY)$ over all joint distributions of
the ordered tuple $\boldX=(\boldX_1,\ldots,\boldX_m)$, where
$\boldX_j\in\{0,1\}^k$ for every $j\in\{1,\ldots,m\}$, and where $\cY$ is the unordered multiset obtained by applying
the BSC to these strings with independent substitution noise across all $mk$ bits.
\end{definition}
If $m=0$, then $\cY$ is the empty multiset, and if $k=0$, then $\cY$
consists of $m$ copies of the empty string. In either case, $\cY$ is
deterministic, and therefore $\Csh(0,k)=\Csh(m,0)=0$.

\begin{remark}
    The quantity $\Csh(m,k)$ is the largest mutual information between the
input and the output of the NSC (Definition~\ref{def:nsc}) with $m$
fragments of length~$k$.
The capacity of the NSC is instead an asymptotic quantity per codeword
bit: it concerns $\Csh(K_n,\ell_n)/n$ as $n\to\infty$, with $\ell_n=\Theta(\log n)$, and Theorem~\ref{thm:nsc} evaluates it.
In contrast, the proof of Theorem~\ref{thm:renewal-capacity} need not follow this scaling on fragment length, so we need bounds on $\Csh(m,k)$ that
hold for every $m$ and $k$.
\end{remark}


\begin{proposition}
\label{prop:uniform}
Fix $\epsilon>0$ and $\delta\in(0,1)$, and define
\begin{equation*}
 b=\frac{1-\delta}{\kappa+\epsilon},
 \qquad\lambda=\kappa b\in(0,1).
\end{equation*}
There is a constant $c_2=c_2(p,\epsilon,\delta)$ such that
for every $m,L\ge1$,
\begin{equation}
 \Csh(m,L)
 \le m(\kappa L-\lambda\log m)^+
       +c_2m+2b\,m^{1-\delta}\log(m+1).
 \label{eq:uniform-shuffling}
\end{equation}
\end{proposition}

Assuming Proposition~\ref{prop:uniform}, we now show how to derive Theorem~\ref{thm:renewal-capacity}.
To obtain an upper bound on the mutual information, we allow the
decoder to know the exact cut locations in the codeword, in addition
to the unordered multiset of noisy fragments.
This can only increase the information available to the decoder, so
the resulting upper bound also applies to the noisy TPC;
see~\eqref{eq:reveal-cuts} in the proof of Lemma~\ref{lem:partition}.
The cut locations determine the length of every fragment, so we can
divide the noisy fragments into sets according to their lengths, and
let $m_L$ denote the number of fragments of length~$L$.
The fragments within each set remain unordered, so
Proposition~\ref{prop:uniform} applies to each set of $m_L$ fragments
of length~$L$.
The proof uses three auxiliary lemmas.
Lemma~\ref{lem:partition-counts} bounds the number of distinct
fragment lengths and controls the additional terms that appear when
the bounds for different lengths are added.
Lemma~\ref{lem:sets} shows that, once the decoder knows the cut
locations, the information in all noisy fragments is at most the sum
of the bounds for the sets of equal-length fragments.
Lemma~\ref{lem:partition} combines these results and averages over the
possible cut locations to obtain a mutual information bound for the
noisy TPC.
Finally, we prove Theorem~\ref{thm:renewal-capacity} under
Assumption~\ref{remark:assumptions}.

We next show that the information carried by all noisy
fragments is at most the sum of the information 
for the fragments of each length. 
\begin{lemma}
\label{lem:sets}
Let $\boldT=(T_1,\ldots,T_{n-1})$ be the cut pattern from Definition~\ref{def:renewal-model}.
Write $\boldt=(t_1,\ldots,t_{n-1})\in\{0,1\}^{n-1}$
for a fixed possible value of $\boldT$.
Fix $\boldt$ with $\Pr(\boldT=\boldt)>0$, and let
$m_L$ be the
number of fragments of length $L$ and
$\cL=\{L\in[n]:m_L>0\}$ as defined in Lemma~\ref{lem:partition-counts} in Appendix~\ref{app:bernoulli}.
Then, for every distribution of $\boldX$,
\begin{align}
I(\boldX;\cY\mid\boldT=\boldt)
&\leq\sum_{L\in\cL}\Csh(m_L,L).
\label{eq:set-information}
\end{align}
\end{lemma}

\begin{proof}
For each $L\in\cL$, let $\cY_L$ be the multiset of
noisy fragments of length $L$. List the distinct
lengths as $L_1<\cdots<L_{|\cL|}$.
Separating $\cY$ by fragment length provides the tuple
$(\cY_L)_{L\in\cL}$, and combining its entries
recovers~$\cY$. Therefore, 
\begin{align}
I(\boldX;\cY\mid\boldT=\boldt)
&=H((\cY_L)_{L\in\cL}\mid\boldT=\boldt)
-H((\cY_L)_{L\in\cL}\mid\boldX,\boldT=\boldt)
\notag\\
&\overset{(a)}{=}
H((\cY_L)_{L\in\cL}\mid\boldT=\boldt)
-\sum_{L\in\cL}H(\cY_L\mid\boldX,\boldT=\boldt)
\notag\\
&\overset{(b)}{=}
\sum_{i=1}^{|\cL|}
H(\cY_{L_i}\mid
\cY_{L_1},\ldots,\cY_{L_{i-1}},\boldT=\boldt)
-\sum_{i=1}^{|\cL|}
H(\cY_{L_i}\mid\boldX,\boldT=\boldt)
\notag\\
&\overset{(c)}{\leq}
\sum_{i=1}^{|\cL|}
\left[
H(\cY_{L_i}\mid\boldT=\boldt)
-H(\cY_{L_i}\mid\boldX,\boldT=\boldt)
\right]
\notag\\
&=\sum_{L\in\cL}
I(\boldX;\cY_L\mid\boldT=\boldt)
\overset{(d)}{\leq}
\sum_{L\in\cL}\Csh(m_L,L),
\notag
\end{align}
where $(a)$ follows from the conditional independence
of the multisets $\cY_L$: once the codeword and cut
locations are fixed, each $\cY_L$ depends only on
the BSC noise affecting the fragments of length $L$.
Fragments of different lengths occupy disjoint bit
positions, so their independent noise makes the
conditional entropies add exactly.
Equality~$(b)$ follows from the entropy chain rule.
Inequality~$(c)$ removes the conditioning
on the earlier multisets, which can only increase
each entropy.
Inequality~$(d)$ follows from
Definition~\ref{definition:Csh}.
\end{proof}
We now apply Proposition~\ref{prop:uniform} to the
fragments of each length and add the resulting bounds.

\begin{lemma}
\label{lem:partition}
Suppose there exists a constant~$c_1>0$ such that
$\Pr(K_n\leq c_1n/\log n)\to1$ as $n\to\infty$. 
Then
\begin{align}
\limsup_{n\to\infty}\frac1n
\sup_{P_{\boldX}}I(\boldX;\cY)
&\leq
\limsup_{n\to\infty}\frac1n
\bE\left[
\sum_{j=1}^{K_n}(\kappa L_{n,j}-\log n)^+
\right].
\label{eq:partition-converse}
\end{align}
\end{lemma}
\begin{proof}
Suppose the decoder knows
the exact cut locations $\boldT$ from
Lemma~\ref{lem:sets}, in addition to $\cY$.
Then
\begin{align}
I(\boldX;\cY)
&\leq I(\boldX;\cY,\boldT)
=I(\boldX;\boldT)+I(\boldX;\cY\mid\boldT)
\overset{(\dagger)}{=}
I(\boldX;\cY\mid\boldT),
\label{eq:reveal-cuts}
\end{align}
where $(\dagger)$ follows from the
independence of $\boldX$ and $\boldT$.

Fix $\epsilon>0$ and $\delta\in(0,1)$, and use the
constants $b,\lambda,c_2$ from
Proposition~\ref{prop:uniform}.
We first consider fixed cut locations
$\boldT=\boldt$ with $K_n\leq c_1n/\log n$
(in calculations conditional on $\boldT=\boldt$,
$K_n$ and $L_{n,j}$ refer to  fixed values
determined by $\boldt$), for which we have,
\begin{align}
I(\boldX;\cY\mid\boldT=\boldt)
&\overset{\eqref{eq:set-information}}{\leq}
\sum_{L\in\cL}\Csh(m_L,L)
\notag\\
&\overset{\eqref{eq:uniform-shuffling}}{\leq}
\sum_{L\in\cL}m_L(\kappa L-\lambda\log m_L)^+
+c_2K_n
+2b\sum_{L\in\cL}m_L^{1-\delta}\log(m_L+1)
\notag\\
&\overset{\eqref{eq:counts-power}}{\leq}
\sum_{L\in\cL}m_L(\kappa L-\lambda\log m_L)^+
+c_2K_n
+2b\log(n+1)(2n)^{\delta/2}K_n^{1-\delta}
\notag\\
&\overset{(\dagger)}{\leq}
\sum_{j=1}^{K_n}(\kappa L_{n,j}-\lambda\log n)^+
+\lambda\sum_{L\in\cL}m_L\log\frac{n}{m_L}
+c_2K_n
+2b\log(n+1)(2n)^{\delta/2}K_n^{1-\delta}
\notag\\
&\overset{\eqref{eq:counts-entropy}}{\leq}
\sum_{j=1}^{K_n}(\kappa L_{n,j}-\lambda\log n)^+
+(c_2+\lambda\log e)K_n
+2\lambda K_n\log\frac{n}{K_n}
+2b\log(n+1)(2n)^{\delta/2}K_n^{1-\delta},
\label{eq:replace-length-logarithms}
\end{align}
where $(\dagger)$ follows since
$\kappa L-\lambda\log m_L
=\kappa L-\lambda\log n+\lambda\log(n/m_L)$
and $(a+d)^+\leq a^++d$ for $d\geq0$.

We next bound the last three terms in~\eqref{eq:replace-length-logarithms},
divided by~$n$, one at a time.
Since $K_n/n\leq c_1/\log n<1/e$ for a sufficiently large~$n$,
we have,
\begin{align*}
\frac{(c_2+\lambda\log e)K_n}{n}
\leq\frac{c_1(c_2+\lambda\log e)}{\log n}.
\end{align*}
For the second term, we cannot bound
$\log(n/K_n)$ by itself; instead we bound the product
$K_n\log(n/K_n)$.
The function $x\mapsto x\log(n/x)$ is nondecreasing on $(0,n/e]$.
Thus, we have
\begin{align*}
\frac{2\lambda K_n\log(n/K_n)}{n}
\leq\frac{2\lambda}{n}\cdot\frac{c_1n}{\log n}
\log\left(\frac{n}{c_1n/\log n}\right)
=\frac{2\lambda c_1\log(\log n/c_1)}{\log n}.
\end{align*}
For the third term, we use $K_n^{1-\delta}\leq(c_1n/\log n)^{1-\delta}$,
which holds since $1-\delta>0$:
\begin{align*}
\frac{2b\log(n+1)(2n)^{\delta/2}K_n^{1-\delta}}{n}
\leq\frac{2b\log(n+1)\,2^{\delta/2}n^{\delta/2}\,c_1^{1-\delta}n^{1-\delta}}
{n(\log n)^{1-\delta}}
=\frac{2^{1+\delta/2}bc_1^{1-\delta}\log(n+1)}
{n^{\delta/2}(\log n)^{1-\delta}}.
\end{align*}
Adding the three bounds shows that the last three terms
in~\eqref{eq:replace-length-logarithms}, divided by~$n$, are at most
\begin{align}
\frac{c_1(c_2+\lambda\log e)}{\log n}
+\frac{2\lambda c_1\log(\log n/c_1)}{\log n}
+\frac{2^{1+\delta/2}bc_1^{1-\delta}\log(n+1)}
{n^{\delta/2}(\log n)^{1-\delta}}
=:e_n.
\notag
\end{align}
Each of the three terms defining $e_n$ tends to
zero as $n\to\infty$, so $e_n\to0$ as $n\to\infty$.

Since~$\lambda\in(0,1)$ (see Proposition~\ref{prop:uniform}), it follows that
\begin{align}
I(\boldX;\cY\mid\boldT=\boldt)
&\overset{\eqref{eq:replace-length-logarithms}}{\leq}
\sum_{j=1}^{K_n}(\kappa L_{n,j}-\lambda\log n)^+
+ne_n
\notag\\
&\overset{(\dagger)}{\leq}
\sum_{j=1}^{K_n}(\kappa L_{n,j}-\log n)^+
+(1-\lambda)K_n\log n+ne_n
\notag\\
&\leq
\sum_{j=1}^{K_n}(\kappa L_{n,j}-\log n)^+
+(1-\lambda)c_1n+ne_n,
\label{eq:partition-small-fragment-number}
\end{align}
where $(\dagger)$ follows since
$\kappa L_{n,j}-\lambda\log n
=\kappa L_{n,j}-\log n+(1-\lambda)\log n$ and since~$(a+d)^+\le a^++d$ for every real~$a$ and every~$d\geq0$; applying this with~$d=(1-\lambda)\log n$ to each of the~$K_n$ terms and summing over~$j$ produces the term~$(1-\lambda)K_n\log n$.

We now average over all cut locations to get the final bound for
$I(\boldX;\cY)$.
\begin{align}
\frac1n I(\boldX;\cY)
&\overset{\eqref{eq:reveal-cuts}}{\leq}
\frac1n I(\boldX;\cY\mid\boldT)\overset{(\dagger)}{=}
\frac1n\sum_{\boldt}
\Pr(\boldT=\boldt)
I(\boldX;\cY\mid\boldT=\boldt)
\notag\\
&\overset{(\dagger\dagger)}{=}
\frac1n\sum_{\boldt}
\Pr(\boldT=\boldt,\,K_n\leq c_1n/\log n)
I(\boldX;\cY\mid\boldT=\boldt)
\notag\\
&\quad+
\frac1n\sum_{\boldt}
\Pr(\boldT=\boldt,\,K_n>c_1n/\log n)
I(\boldX;\cY\mid\boldT=\boldt),
\label{eq:split-cut-patterns}
\end{align}
where $(\dagger)$ follows from the definition of conditional
mutual information, and $(\dagger\dagger)$ follows from the law of
total probability by separating the cut locations according to
whether $K_n\leq c_1n/\log n$ or $K_n>c_1n/\log n$.

For the cut patterns with $K_n\leq c_1n/\log n$, we bound
$I(\boldX;\cY\mid\boldT=\boldt)$
using~\eqref{eq:partition-small-fragment-number}.
For the remaining cut patterns, we use the trivial bound
$I(\boldX;\cY\mid\boldT=\boldt)\leq H(\boldX)\leq n$.
Given $\boldT=\boldt$, the number of fragments $K_n$ and the lengths
$L_{n,j}$ are fixed, so averaging the first bound over $\boldt$ turns
its sum into an expectation.
Substituting these bounds into~\eqref{eq:split-cut-patterns} shows
\begin{align}
\frac1n I(\boldX;\cY)
&\leq
\frac1n\bE\left[
\mathbf{1}_{\{K_n\leq c_1n/\log n\}}
\sum_{j=1}^{K_n}(\kappa L_{n,j}-\log n)^+
\right]
\notag\\
&\quad+\bigl((1-\lambda)c_1+e_n\bigr)\Pr(K_n\leq c_1n/\log n)
+\Pr(K_n>c_1n/\log n)
\notag\\
&\overset{(\dagger)}{\leq}
\frac1n\bE\left[
\sum_{j=1}^{K_n}(\kappa L_{n,j}-\log n)^+
\right]
+(1-\lambda)c_1+e_n
+\Pr(K_n>c_1n/\log n),
\label{eq:partition-final-finite}
\end{align}
where $(\dagger)$ drops the indicator, which can only increase the
expectation since every summand is nonnegative, and uses
$\Pr(K_n\leq c_1n/\log n)\leq1$.

Now we prove the main statement~\eqref{eq:partition-converse}:
\begin{align}
\limsup_{n\to\infty}\frac1n
\sup_{P_{\boldX}}I(\boldX;\cY)
&\overset{\eqref{eq:partition-final-finite}}{\leq}
\lim_{\epsilon,\delta\to0}\limsup_{n\to\infty}
\left\{
\frac1n\bE\left[
\sum_{j=1}^{K_n}(\kappa L_{n,j}-\log n)^+
\right]
+(1-\lambda)c_1+e_n
+\Pr(K_n>c_1n/\log n)
\right\}
\notag\\
&\overset{(\dagger)}{=}
\lim_{\epsilon,\delta\to0}
\left\{
\limsup_{n\to\infty}\frac1n
\bE\left[
\sum_{j=1}^{K_n}(\kappa L_{n,j}-\log n)^+
\right]
+c_1\left(1-\frac{\kappa(1-\delta)}
{\kappa+\epsilon}\right)
\right\}
\notag\\
&\overset{(\dagger\dagger)}{=}
\limsup_{n\to\infty}\frac1n
\bE\left[
\sum_{j=1}^{K_n}(\kappa L_{n,j}-\log n)^+
\right],
\notag
\end{align}
where $(\dagger)$ follows since $e_n\to0$ and
$\Pr(K_n>c_1n/\log n)\to0$ as $n\to\infty$
for each fixed $\epsilon,\delta$.
Equality $(\dagger\dagger)$ follows from
$\kappa(1-\delta)/(\kappa+\epsilon)\to1$
as $\epsilon,\delta\to0$. \qedhere
\end{proof}
Lemma~\ref{lem:partition} bounds the mutual information
for the noisy TPC. 
We now use
the assumptions in~\eqref{eq:renewal-assumptions}
to prove Theorem~\ref{thm:renewal-capacity}.

\begin{proof}[Proof of Theorem~\ref{thm:renewal-capacity}]

Choose a fixed $c_1>\alpha$. We have
\begin{align}
\frac1n\bE\left[
\sum_{1\leq j\leq c_1n/\log n}N_{n,j}
\right]
&=\frac1n\sum_{1\leq j\leq c_1n/\log n}
\bE[N_{n,j}]=\frac{\ell_n}{n}
\sum_{1\leq j\leq c_1n/\log n}1\notag\\
&\overset{}{=}
\frac{\ell_n}{\log n}c_1+O(\ell_n/n)
\overset{\eqref{eq:renewal-assumptions}}{\longrightarrow}
\frac{c_1}{\alpha}>1
\qquad(n\to\infty).
\label{eq:convergence of mean}
\end{align}

Then, for all sufficiently large $n$, we have
\begin{align}
\Pr(K_n>c_1n/\log n)
&\overset{(a)}{=}
\Pr\left(
\sum_{1\leq j\leq c_1n/\log n}N_{n,j}<n
\right)
\notag\\
&\overset{(b)}{\leq}
\frac{
\Var\left(\sum_{1\leq j\leq c_1n/\log n}N_{n,j}\right)}
{\left(
\bE\left[\sum_{1\leq j\leq c_1n/\log n}N_{n,j}\right]-n
\right)^2}
\notag\\
&\overset{(c)}{\leq}
\frac{
4\Var\left(\sum_{1\leq j\leq c_1n/\log n}N_{n,j}\right)}
{(c_1/\alpha-1)^2n^2}
\notag\\
&\overset{(d)}{\leq}
\frac{4c_1\bE[N_{n,1}^2]}
{(c_1/\alpha-1)^2n\log n}
\overset{\eqref{eq:renewal-assumptions}}{=}
O\left(\frac{\log n}{n}\right)
\longrightarrow0
\qquad(n\to\infty),
\notag
\end{align}
where $(a)$ follows since more than $c_1n/\log n$
fragments are needed exactly when the sampled
lengths in the sum do not reach $n$.
Inequality~$(b)$ follows from Chebyshev's
inequality. When the sum is less than $n$,
its distance from its mean is greater than
the mean minus $n$.
Inequality~$(c)$ follows from~\eqref{eq:convergence of mean}.
In particular,
\begin{align*}
\frac1n\bE\left[
\sum_{1\leq j\leq c_1n/\log n}N_{n,j}
\right]
&\geq
\frac{c_1}{\alpha}
-\frac12\left(\frac{c_1}{\alpha}-1\right)
=1+\frac12\left(\frac{c_1}{\alpha}-1\right).
\end{align*}
Inequality~$(d)$ follows from the
independence and identical distribution of the
sampled lengths:
independence allows us to add
their variances:
\begin{align*}
\Var\left(
\sum_{1\leq j\leq c_1n/\log n}N_{n,j}
\right)
&=\sum_{1\leq j\leq c_1n/\log n}\Var(N_{n,j})\\
&=\sum_{1\leq j\leq c_1n/\log n}
\bigl(\bE[N_{n,1}^2]-\ell_n^2\bigr)
\leq
\frac{c_1n}{\log n}\bE[N_{n,1}^2].
\end{align*}

We have shown that $\Pr(K_n\leq c_1n/\log n)\to1$, so
Lemma~\ref{lem:partition} applies, and it remains to bound the expected
sum in~\eqref{eq:partition-converse}.
\begin{align}
\bE\left[
\sum_{j=1}^{K_n}(\kappa N_{n,j}-\log n)^+
\right]
&\overset{(\dagger)}{=}\sum_{j=1}^{n}\Pr(K_n\geq j)
\bE[(\kappa N_{n,1}-\log n)^+]
\notag\\
&\overset{\eqref{eq:reward-rate}}{=}
r_n(\kappa)\sum_{j=1}^{n}\Pr(K_n\geq j)\ell_n
\overset{(\dagger\dagger)}{=}
r_n(\kappa)\bE\left[\sum_{j=1}^{K_n}N_{n,j}\right],
\label{eq:sampled-fragment-expectation}
\end{align}
where $(\dagger)$ and $(\dagger\dagger)$ follow from Wald's
equation~\cite{wald1944cumulative} (see also~\cite{wikiWald}), applied to the i.i.d.\ sequences
$((\kappa N_{n,j}-\log n)^+)_{j\geq1}$ and $(N_{n,j})_{j\geq1}$,
respectively.
It applies since $K_n\leq n$ and the event
$\{K_n\geq j\}=\{S_{n,j-1}<n\}$ is determined by
$N_{n,1},\ldots,N_{n,j-1}$, and is therefore independent of $N_{n,j}$.

Every actual fragment length satisfies
$L_{n,j}\leq N_{n,j}$,
which provides
\begin{align}
\frac1n\bE\left[
\sum_{j=1}^{K_n}(\kappa L_{n,j}-\log n)^+
\right]
&\leq
\frac1n\bE\left[
\sum_{j=1}^{K_n}(\kappa N_{n,j}-\log n)^+
\right]
\notag\\
&\overset{\eqref{eq:sampled-fragment-expectation}}{=}
\frac{r_n(\kappa)}{n}
\bE\left[\sum_{j=1}^{K_n}N_{n,j}\right]
\notag\\
&\overset{(\dagger)}{\leq}
r_n(\kappa)\left(
1+\sqrt{\frac{\bE[N_{n,1}^2]}{n}}
\right)
\overset{\eqref{eq:renewal-assumptions}}{=}
r_n(\kappa)+o(1)
\qquad(n\to\infty),
\label{eq:expected-fragment-upper}
\end{align}
where $(\dagger)$ follows since
$\sum_{j=1}^{K_n}N_{n,j}
\leq n+\max_{1\leq j\leq n}N_{n,j}$
and
$\bE[\max_{1\leq j\leq n}N_{n,j}]
\leq\sqrt{\bE[\sum_{j=1}^{n}N_{n,j}^2]}
=\sqrt{n\bE[N_{n,1}^2]}$.

Applying
Lemma~\ref{lem:partition} therefore implies
\begin{align*}
\limsup_{n\to\infty}\frac1n
\sup_{P_{\boldX}}I(\boldX;\cY)
&\overset{\eqref{eq:partition-converse}}{\leq}
\limsup_{n\to\infty}\frac1n
\bE\left[
\sum_{j=1}^{K_n}(\kappa L_{n,j}-\log n)^+
\right]
\overset{\eqref{eq:expected-fragment-upper}}{\leq}
\limsup_{n\to\infty}r_n(\kappa). \qedhere
\end{align*}
\end{proof}

Next, we apply Theorem~\ref{thm:renewal-capacity} to
the independent-cut torn paper channel with
substitution errors. 
We evaluate
$r_n(\kappa)$ for this distribution to obtain the
capacity upper bound. 
The upper bound matches the
achievable rate in~\cite[Theorem~3]{ravi2024recovering},
which determines the capacity and proves
Corollary~\ref{cor:geometric}.
\begin{proof}[Proof of Corollary~\ref{cor:geometric}]
We first verify that the independent-cut torn paper channel
satisfies the fragment-length assumptions in
Assumption~\ref{remark:assumptions}.
Ravi, Vahid, and Shomorony describe the original torn paper
channel in~\cite[Section~I]{ravi2024recovering} as the special
case in which the fragment lengths are geometrically distributed: every possible cut
between two consecutive bits occurs independently with probability
$q_n$.
So, we have
$\Pr(N_{n,1}=j)=q_n(1-q_n)^{j-1},\ j\geq1,
$
and
$\ell_n=\bE[N_{n,1}]=\frac1{q_n},\bE[N_{n,1}^2]=\frac{2-q_n}{q_n^2}.$
The assumption $q_n\log n\to\alpha\in(0,\infty)$ provides
$\frac{\log n}{\ell_n}=q_n\log n\longrightarrow\alpha$ when $n\to\infty$.
Moreover, $q_n\log n\to\alpha$ when $n\to\infty$ implies
$q_n^{-1}=\Theta(\log n)$, and hence
$\bE[N_{n,1}^2]=\frac{2-q_n}{q_n^2}=\Theta((\log n)^2)$ when $n\to\infty$.
Thus, the independent-cut torn paper channel satisfies
\eqref{eq:renewal-assumptions}.

Next, we prove an explicit bound for $r_n(\kappa)$.
\begin{align}
r_n(\kappa)
&\overset{\eqref{eq:reward-rate}}{=}
q_n\bE[(\kappa N_{n,1}-\log n)^+]
=\kappa q_n\bE\left[\left(N_{n,1}-\frac{\log n}{\kappa}\right)^+\right]
\overset{(\dagger)}{=}
\kappa(1-q_n)^{\log n/\kappa}
\notag\\
&=
\kappa\exp\Bigl(\frac{q_n\log n}{\kappa}\cdot\frac{\ln(1-q_n)}{q_n}\Bigr)
\overset{(\dagger\dagger)}{\longrightarrow}
\kappa e^{-\alpha/\kappa},\qquad n\to\infty,
\label{eq:geometric-reward}
\end{align}
where, for simplicity, we assume that $\log n/\kappa$ is an integer.
Equality~$(\dagger)$ follows from the memoryless property of the
geometric distribution.
The limit~$(\dagger\dagger)$ follows from $q_n\log n\to\alpha$ and
$\ln(1-q_n)/q_n\to-1$, where $q_n\to0$ since $q_n\log n\to\alpha$.
If $\log n/\kappa$ is not an integer, let $a$ be the integer with
$a\leq\log n/\kappa<a+1$.
Since $\bE[(N_{n,1}-x)^+]$ is nonincreasing in $x$, the same
argument shows that $\kappa(1-q_n)^{a+1}\leq r_n(\kappa)\leq\kappa(1-q_n)^{a}$,
and both bounds tend to $\kappa e^{-\alpha/\kappa}$ since
$(1-q_n)^{\pm1}\to1$ when $n\to\infty$.

Theorem~\ref{thm:renewal-capacity} and Fano's inequality, as in the proof of Corollary~\ref{cor:capacity-limit}, therefore provide
\begin{align*}
C(\alpha,p)
&\leq\limsup_{n\to\infty}r_n(\kappa)
\overset{\eqref{eq:geometric-reward}}{=}
\kappa e^{-\alpha/\kappa}.
\end{align*}
On the other hand, the achievable
rate in~\cite[Theorem~3]{ravi2024recovering} provides
$
C(\alpha,p)
\geq
\kappa e^{-\alpha/\kappa}.
$
Combining the upper and lower bounds yields
$
C(\alpha,p)
=
\kappa e^{-\alpha/\kappa},
$
which proves~\eqref{eq:geometric-capacity}.
\end{proof}

\section{Capacity of the noisy shuffling channel}
\label{sec:sparse}
\label{sec:suffix}
In this section, we prove Theorem~\ref{thm:nsc}, which determines the capacity of the NSC. 
Since the NSC consists of equal-length noisy
fragments, the main step is Proposition~\ref{prop:uniform}, which bounds the information carried by an unordered collection of such fragments.
To this end, we first construct a finite collection of distributions (Construction~\ref{con:sparse-family}) that can approximate the distribution of any single noisy
string obtained from the BSC. 
Every member of the finite collection describes
a distribution of noisy strings.
Lemma~\ref{lem:sparse} shows that the actual
distribution of a noisy string can be approximated
by one member of this collection, regardless of
the probabilities of the binary strings before
the BSC. The selected member may depend on those
probabilities, but the collection itself is fixed.
We then use this collection,
together with a bound accounting for dependence among randomly
ordered noisy strings (Lemma~\ref{lem:type}), to control the information carried by the entire multiset in Proposition~\ref{prop:short}.
The bound in Proposition~\ref{prop:short} is useful when the fragments are short.
For longer fragments, we divide each fragment into a prefix and a
suffix by choosing the prefix length so that the short-fragment
bound applies.
We sort the noisy prefixes and retain their
corresponding noisy suffixes in the same order, and then bound the additional information carried by the suffixes, including any help they provide in identifying the original positions
of the prefixes. 
Combining the prefix and suffix bounds provides
an estimate valid for every number of fragments and every
fragment length as in Lemma~\ref{lem:suffix}. 
This estimate yields Proposition~\ref{prop:uniform}, and applying
Proposition~\ref{prop:uniform} to the equal-length fragments of the NSC
proves Theorem~\ref{thm:nsc} at the end of this section.

For two binary strings $\boldx,\boldy\in\{0,1\}^k$, let
$W_k(\boldy\mid\boldx)$ denote the probability that applying
the BSC to $\boldx$ provides the noisy string $\boldy$. Since each bit flips
independently with probability $p$, we have
\begin{equation}
 W_k(\boldy\mid\boldx)
 =p^{d_H(\boldx,\boldy)}(1-p)^{k-d_H(\boldx,\boldy)},
 \label{eq:bsc-kernel}
\end{equation}
where $d_H(\boldx,\boldy)$ is the Hamming distance between
$\boldx$ and $\boldy$.

We next seek an upper bound for $\Csh(m,k)$ that accounts
for the loss of fragment order. As a first step, fix a
fragment length $k\geq1$ and consider any distribution $P$ 
on the binary strings $\{0,1\}^k$ before the BSC (Remark~\ref{remark:distribution before bsc}).
The
corresponding distribution of one noisy string is denoted
by $\widetilde{P}$ ; that is, $\widetilde{P}(\boldy)$ is the probability of obtaining $\boldy$ at the output of the BSC when the input string is chosen according to~$P$, where
\[
\widetilde{P}(\boldy)
=
\sum_{\boldx\in\{0,1\}^k}
P(\boldx)W_k(\boldy\mid\boldx),
\qquad
\boldy\in\{0,1\}^k.
\]
Our immediate goal is to construct a finite collection
$\cV_k$ of distributions on noisy strings such that every
possible $\widetilde{P}$ can be approximated by some member of
$\cV_k$. 
The collection must be chosen independently of
$P$, although the member used to approximate $\widetilde{P}$
may depend on $P$.

\begin{remark}\label{remark:distribution before bsc}
In what follows, we treat each fixed tuple of equal-length binary strings, that is, each possible value of the strings before the BSC, as a probability distribution on single strings, namely its empirical distribution.
Accordingly, the distribution of strings before the BSC considered above will be taken to be the empirical distribution of a fixed tuple, rather than the distribution of a random tuple.
For a fixed tuple of binary strings, the empirical
distribution is fixed: the probability assigned to
each string is its number of occurrences divided
by the number of entries in the tuple.
Selecting one entry uniformly therefore provides one
string with this distribution, and applying the BSC
to the selected string provides the distribution $\widetilde{P}$.
We prove the approximation for any fixed distribution
$P$, so that it can later be applied to the empirical
distribution of every fixed tuple. 
Different tuples
may provide different empirical distributions, but the
same finite collection must provide an approximation
for all of them.
This empirical distribution is defined formally in Definition~\ref{def:fixed-fragment-output} and is used in this way in the proof of Proposition~\ref{prop:short}.
\end{remark}

\begin{construction}
\label{con:sparse-family}
Fix $\epsilon>0$, and for each integer $k\geq1$, let
$J_k=\lceil2^{k(\kappa+\epsilon)}\rceil$. For every ordered list
$\boldx=(\boldx_1,\ldots,\boldx_{J_k})$ of binary strings with
$\boldx_i\in\{0,1\}^k$ for all $i\in [J_k]=\{1,2,\ldots,J_k\}$, define 
\[
 Q_{\boldx}(\boldy)=\frac1{J_k}\sum_{j=1}^{J_k}
 W_k(\boldy\mid\boldx_j),
 \qquad \boldy\in\{0,1\}^k.
\]
Let $\cV_k$ be the collection of all distinct distributions
$Q_{\boldx}$ obtained from these lists.
\end{construction}

Each distribution in $\cV_k$ describes the noisy string obtained
by choosing one position in its defining list uniformly and
passing the corresponding string through the BSC.
There are $2^k$ choices for each of the $J_k$ positions in
a list, yielding $2^{kJ_k}$ possible lists. 
Since different
lists may define the same distribution, $|\cV_k|\leq 2^{kJ_k}$.
Lemma~\ref{lem:sparse} below shows that this collection contains an approximation
for all distributions of the form \(\widetilde{P}\) produced by the BSC. 
The next lemma shows that, for a fixed input string, the noisy strings
with an atypical number of bit flips have small total probability,
while every noisy string with a typical number of bit flips has
probability at most $2^{k(\kappa+\epsilon/2)}2^{-k}$.
\begin{lemma}
\label{lem:truncated-bsc}
Fix $\epsilon>0$ and choose $s$ with $0<s<p$ and
$s\log((1-p)/p)\leq\epsilon/2$.
 For $\boldx,\boldy\in\{0,1\}^k$, define
\begin{align}
 W_k^\circ(\boldy\mid\boldx)
 =W_k(\boldy\mid\boldx)
 \mathbf{1}_{\{|d_H(\boldx,\boldy)-kp|\leq ks\}}, \qquad
 \tau_k=1-\sum_{\boldy\in\{0,1\}^k}
 W_k^\circ(\boldy\mid\boldx).
 \label{eq:tau_k}
\end{align}
Then $\tau_k$  does not depend on \(\mathbf{x}\) and satisfies
$\tau_k\leq2e^{-2ks^2}$. 
Moreover, for every
$\boldx,\boldy\in\{0,1\}^k$ we have $W_k^\circ(\boldy\mid\boldx)
\leq2^{k(\kappa+\epsilon/2)}2^{-k}$.
\end{lemma}

\begin{proof}
We first verify that the required choice of $s$ exists.
Since $p\in(0,1/2)$,
the quantity $\log((1-p)/p)$ is positive and finite, so
sufficiently small $s>0$ exists to satisfy $0<s<p$ and
$s\log((1-p)/p)\leq\epsilon/2$.
Choose such $s$ using only $p$ and $\epsilon$, and keep it fixed as $k$ varies.

To calculate $\tau_k$, fix $\boldx\in\{0,1\}^k$ and let $Z$ denote the number
of bits flipped by the BSC. Each of the $k$ bits is flipped
independently with probability $p$, so $Z$ has the
Binomial$(k,p)$ distribution.
By definition, $W_k^\circ(\boldy\mid\boldx)$ keeps the original
probability $W_k(\boldy\mid\boldx)$ when
$|d_H(\boldx,\boldy)-kp|\leq ks$ and sets it to zero otherwise.
Thus, $\tau_k$ is exactly the probability that the number of
flips lies outside this range: 
\begin{align*}
 \tau_k
 =1-\sum_{\boldy\in\{0,1\}^k}
       W_k^\circ(\boldy\mid\boldx) =\sum_{\substack{\boldy\in\{0,1\}^k\\
                   |d_H(\boldx,\boldy)-kp|>ks}}
       W_k(\boldy\mid\boldx) =\Pr(|Z-kp|>ks).
\end{align*}
Applying 
Lemma~\ref{lem:bernoulli} in Appendix~\ref{app:bernoulli}  yields
\[
 \tau_k=\Pr(|Z-kp|>ks)\leq 2e^{-2ks^2}.
\]

To bound the probability of each retained noisy string, fix
$\boldx,\boldy\in\{0,1\}^k$. Taking the negative logarithm
of~\eqref{eq:bsc-kernel}  provides
\begin{align*}
-\log W_k(\boldy\mid\boldx)
&=-d_H(\boldx,\boldy)\log p
-(k-d_H(\boldx,\boldy))\log(1-p) \\
&=-d_H(\boldx,\boldy)\log p
-k\log(1-p)
+d_H(\boldx,\boldy)\log(1-p) \\
&=-k\log(1-p)
+d_H(\boldx,\boldy)\log\frac{1-p}{p} \\
&=-kp\log p-k(1-p)\log(1-p) 
+\bigl(d_H(\boldx,\boldy)-kp\bigr)
\log\frac{1-p}{p} \\
&=kh_2(p)
+\bigl(d_H(\boldx,\boldy)-kp\bigr)
\log\frac{1-p}{p}.
\end{align*}
For the case $|d_H(\boldx,\boldy)-kp|\leq ks$, $W_k^\circ(\boldy\mid\boldx)=W_k(\boldy\mid\boldx)>0$, so 
\begin{align*}
    -\log W^\circ_k(\boldy\mid\boldx)\overset{(\dagger)}\geq kh_2(p)-ks\log\frac{1-p}{p} 
\overset{(\dagger\dagger)}\geq k(h_2(p)-\frac{\epsilon}{2})
\end{align*}
 where $(\dagger)$ follows from $|d_H(\boldx,\boldy)-kp|\leq ks$, which implies $d_H(\boldx,\boldy)-kp\geq-ks$ and $(\dagger\dagger)$ follows from
$s\log((1-p)/p)\leq\epsilon/2$.

Thus, $W^\circ_k(\boldy\mid\boldx)\leq
2^{-k(h_2(p)-\epsilon/2)}$.
Using $\kappa=1-h_2(p)$
provides
\begin{align}
 W_k^\circ(\boldy\mid\boldx)
 &\leq 2^{-k(h_2(p)-\epsilon/2)}
 =2^{k(\kappa+\epsilon/2)}2^{-k}.
 \label{eq:typical-row-bound}
\end{align}
For all $\boldx,\boldy$ satisfying $|d_H(\boldx,\boldy)-kp|>ks$,
we have $W_k^\circ(\boldy\mid\boldx)=0$,
so~\eqref{eq:typical-row-bound} holds for every
$\boldx,\boldy\in\{0,1\}^k$.
\end{proof}

Now fix any distribution $P$ on $\{0,1\}^k$ and let
$\widetilde{P}^\circ(\boldy)=
\sum_{\boldx\in\{0,1\}^k}P(\boldx)W_k^\circ(\boldy\mid\boldx)$.
We now introduce auxiliary random strings
$\boldX_1,\ldots,\boldX_{J_k}$, each drawn independently from~$P$.
For every $\boldy\in\{0,1\}^k$, define
\begin{align*}
 \hat{P}(\boldy)
 &=\frac1{J_k}\sum_{j=1}^{J_k}W_k(\boldy\mid\boldX_j) \qquad
 \hat{P}^\circ(\boldy)
 =\frac1{J_k}\sum_{j=1}^{J_k}
 W_k^\circ(\boldy\mid\boldX_j).
\end{align*}
 In the proofs below, $\bE$ and
$\Var$ are taken over $\boldX_1,\ldots,\boldX_{J_k}$, the auxiliary random strings above, each drawn independently from~$P$. 


\begin{lemma}
\label{lem:sampled-mixture-variance}
For every $\boldy\in\{0,1\}^k$, we have
\begin{align*}
 \bE[\hat{P}^\circ(\boldy)]
 =\widetilde{P}^\circ(\boldy) \mbox{ and }
 \Var(\hat{P}^\circ(\boldy))
 \leq\frac{2^{k(\kappa+\epsilon/2)}}{J_k}
 2^{-k}\widetilde{P}^\circ(\boldy).
\end{align*}
\end{lemma}

\begin{proof}
To bound the variance at a fixed noisy string $\boldy$, we
first compute the expectation of $\hat{P}^\circ(\boldy)$.
Each auxiliary sampled string $\boldX_j$ has distribution $P$,
so averaging $W_k^\circ(\boldy\mid\boldX_j)$ over its
possible values provides
\begin{align}\label{eq:expectation of Wk circ}
 \bE[W_k^\circ(\boldy\mid\boldX_j)]
 &=\sum_{\boldx\in\{0,1\}^k}
 P(\boldx)W_k^\circ(\boldy\mid\boldx)
 =\widetilde{P}^\circ(\boldy).
\end{align}
By the definition of $\hat{P}^\circ(\boldy)$ and linearity
of expectation,
\begin{align*}
 \bE[\hat{P}^\circ(\boldy)]
 =\bE\left[\frac1{J_k}\sum_{j=1}^{J_k}
 W_k^\circ(\boldy\mid\boldX_j)\right] 
 =\frac1{J_k}\sum_{j=1}^{J_k}
 \bE[W_k^\circ(\boldy\mid\boldX_j)] 
 =\frac1{J_k}\sum_{j=1}^{J_k}\widetilde{P}^\circ(\boldy)
 =\widetilde{P}^\circ(\boldy).
\end{align*}
We next compute the variance of $\hat{P}^\circ(\boldy)$. 
Each
$W_k^\circ(\boldy\mid\boldX_j)$ depends only on $\boldX_j$.
Since $\boldX_1,\ldots,\boldX_{J_k}$ are independent and
identically distributed, the random variables $W_k^\circ(\boldy\mid\boldX_j)$, $j\in[J_k]$ are also independent
and identically distributed. 
For distinct indices $i\neq j$, independence therefore provides
\begin{align}\label{eq:distinct indice expectation}
 \bE&\left[
 \bigl(W_k^\circ(\boldy\mid\boldX_i)-\widetilde{P}^\circ(\boldy)\bigr)
 \bigl(W_k^\circ(\boldy\mid\boldX_j)-\widetilde{P}^\circ(\boldy)\bigr)
 \right]\notag\\ &=
 \bE\left[W_k^\circ(\boldy\mid\boldX_i)-\widetilde{P}^\circ(\boldy)\right]
 \bE\left[W_k^\circ(\boldy\mid\boldX_j)-\widetilde{P}^\circ(\boldy)\right]
 =0.
\end{align}
Thus, we have
\begin{align}
 \Var(\hat{P}^\circ(\boldy))
 &=\frac1{J_k^2}\bE\left[
 \left(\sum_{j=1}^{J_k}
 \bigl(W_k^\circ(\boldy\mid\boldX_j)-\widetilde{P}^\circ(\boldy)\bigr)
 \right)^2\right] \notag\\
 &=\frac1{J_k^2}\sum_{j=1}^{J_k}
 \bE\left[
 \bigl(W_k^\circ(\boldy\mid\boldX_j)-\widetilde{P}^\circ(\boldy)\bigr)^2
 \right]\notag\\&+\frac{2}{J_k^2}\sum_{1\leq i<j\leq J_k}
 \bE\left[
 \bigl(W_k^\circ(\boldy\mid\boldX_i)-\widetilde{P}^\circ(\boldy)\bigr)
 \bigl(W_k^\circ(\boldy\mid\boldX_j)-\widetilde{P}^\circ(\boldy)\bigr)
 \right]\notag\\
 &\overset{\eqref{eq:distinct indice expectation}}=\frac1{J_k^2}\sum_{j=1}^{J_k}
 \bE\left[
 \bigl(W_k^\circ(\boldy\mid\boldX_j)-\widetilde{P}^\circ(\boldy)\bigr)^2
 \right] \notag\\&=\frac1{J_k}\Var(W_k^\circ(\boldy\mid\boldX_1))=\frac1{J_k}\bigl(\bE[W_k^\circ(\boldy\mid\boldX_1)^2]
   -(\widetilde{P}^\circ(\boldy))^2\bigr)\notag\\
 &\leq\frac1{J_k} \bE[W_k^\circ(\boldy\mid\boldX_1)^2]\overset{(\dagger)}\leq\frac{2^{k(\kappa+\epsilon/2)}}{J_k}
 2^{-k}\bE[W_k^\circ(\boldy\mid\boldX_1)]
 \overset{(\dagger\dagger)}=\frac{2^{k(\kappa+\epsilon/2)}}{J_k}
 2^{-k}\widetilde{P}^\circ(\boldy). \label{eq:variance of p hat circ}
\end{align}
where $(\dagger)$ follows since~\eqref{eq:typical-row-bound} holds for every binary
string before the BSC, so
$0\leq W_k^\circ(\boldy\mid\boldX_1)
\leq2^{k(\kappa+\epsilon/2)}2^{-k}$ and $(\dagger\dagger)$ follows from~\eqref{eq:expectation of Wk circ}.
\end{proof}

The next lemma shows that the collection in Construction~\ref{con:sparse-family} contains a distribution close to the distribution of one noisy string. Lemma~\ref{lem:sparse} then turns this closeness into a relative-entropy bound.
\begin{lemma}
\label{lem:fixed-mixture-approximation}
For any real-valued function $f$ on $\{0,1\}^k$, define
$\|f\|_1=\sum_{\boldy\in\{0,1\}^k}|f(\boldy)|$.
Then, for every distribution $P$ on $\{0,1\}^k$, we have
$\bE[\|\hat{P}-\widetilde{P}\|_1]\leq2\tau_k+2^{-k\epsilon/4}$, where the expectation is over the strings
$\boldX_1,\ldots,\boldX_{J_k}$.
Consequently, there exists $Q\in\cV_k$ with
$\|\widetilde{P}-Q\|_1\leq2\tau_k+2^{-k\epsilon/4}$ and
$\frac12\|\widetilde{P}-Q\|_1\leq v_k$, where
\begin{equation}
 v_k=\min\left\{1,\,
 2e^{-2ks^2}+2^{-k\epsilon/4-1}\right\}.
 \label{eq:tv-sequence}
\end{equation}
The sequence $(v_k)$ with $s$ chosen as in Lemma~\ref{lem:truncated-bsc} depends only on $p,\epsilon$ and tends
to zero as $k\to\infty$.
\end{lemma}
\begin{proof}

We first bound the expected absolute difference at a fixed
$\boldy$. Applying Cauchy--Schwarz to
$|\hat{P}^\circ(\boldy)-\widetilde{P}^\circ(\boldy)|$ provides
\begin{align*}
 \bE\left[
 |\hat{P}^\circ(\boldy)-\widetilde{P}^\circ(\boldy)|
 \right]
 &\leq
 \sqrt{\bE\left[
 (\hat{P}^\circ(\boldy)-\widetilde{P}^\circ(\boldy))^2
 \right]\bE[1^2]} =\sqrt{\Var(\hat{P}^\circ(\boldy))}
 \overset{\eqref{eq:variance of p hat circ}}\leq
 \sqrt{\frac{2^{k(\kappa+\epsilon/2)}}{J_k}}
 \sqrt{2^{-k}\widetilde{P}^\circ(\boldy)}.
\end{align*}
Summing this bound over all possible noisy strings provides
\begin{align}
 \bE\left[\|\hat{P}^\circ-\widetilde{P}^\circ\|_1\right]
 &=\sum_{\boldy\in\{0,1\}^k}
 \bE\left[
 |\hat{P}^\circ(\boldy)-\widetilde{P}^\circ(\boldy)|
 \right] \leq
 \sqrt{\frac{2^{k(\kappa+\epsilon/2)}}{J_k}}
 \sum_{\boldy\in\{0,1\}^k}
 \sqrt{2^{-k}\widetilde{P}^\circ(\boldy)}\notag\\
 &\overset{(\dagger)}\leq
 \sqrt{\frac{2^{k(\kappa+\epsilon/2)}}{J_k}}\sqrt{
 \left(\sum_{\boldy\in\{0,1\}^k}2^{-k}\right)
 \left(\sum_{\boldy\in\{0,1\}^k}\widetilde{P}^\circ(\boldy)\right)}
 \notag\\
 &\overset{(\dagger\dagger)}\leq
 \sqrt{\frac{2^{k(\kappa+\epsilon/2)}}{J_k}}
 \sqrt{1-\tau_k} \overset{(\dagger\dagger\dagger)}\leq
 \sqrt{\frac{2^{k(\kappa+\epsilon/2)}}{2^{k(\kappa+\epsilon)}}}
 =2^{-k\epsilon/4}.
 \label{eq:mixture-l1}
\end{align}
Here, $(\dagger)$ follows from applying Cauchy--Schwarz to the
two finite lists
$\{\sqrt{2^{-k}}\}_{\boldy\in\{0,1\}^k}$ and
$\{\sqrt{\widetilde{P}^\circ(\boldy)}\}_{\boldy\in\{0,1\}^k}$;
$(\dagger\dagger)$ follows from
$\sum_{\boldy\in\{0,1\}^k}2^{-k}=1$ and $\sum_{\boldy\in\{0,1\}^k}\widetilde{P}^\circ(\boldy)\overset{\eqref{eq:expectation of Wk circ}}=\sum_{\boldy\in\{0,1\}^k}\bE[W_k^\circ(\boldy\mid\boldX_j)]\overset{\eqref{eq:tau_k}}=1-\tau_k$;
$(\dagger\dagger\dagger)$ follows from $1-\tau_k\leq1$ and
$J_k\geq2^{k(\kappa+\epsilon)}$ in Construction~\ref{con:sparse-family}.

For every binary string $\boldx$, the difference
$W_k(\boldy\mid\boldx)-W_k^\circ(\boldy\mid\boldx)$ is
nonnegative for every $\boldy$, and its sum over all possible
noisy strings equals $\tau_k$. Hence,
$\widetilde{P}(\boldy)-\widetilde{P}^\circ(\boldy)\geq0$:
\begin{align}
 \|\widetilde{P}-\widetilde{P}^\circ\|_1
 &=\sum_{\boldy\in\{0,1\}^k}
 \bigl(\widetilde{P}(\boldy)-\widetilde{P}^\circ(\boldy)\bigr) \notag\\
 &=\sum_{\boldy\in\{0,1\}^k}
 \sum_{\boldx\in\{0,1\}^k}P(\boldx)
 \bigl(W_k(\boldy\mid\boldx)
       -W_k^\circ(\boldy\mid\boldx)\bigr) \notag\\
 &=\sum_{\boldx\in\{0,1\}^k}P(\boldx)
 \sum_{\boldy\in\{0,1\}^k}
 \bigl(W_k(\boldy\mid\boldx)
       -W_k^\circ(\boldy\mid\boldx)\bigr) \notag\\
 &\overset{\eqref{eq:tau_k}}=\sum_{\boldx\in\{0,1\}^k}P(\boldx)\tau_k
 =\tau_k.
\label{eq:l1 norm between mu and mu circ}
\end{align}

For every 
$\boldX_1,\ldots,\boldX_{J_k}$, each difference
$W_k(\boldy\mid\boldX_j)-W_k^\circ(\boldy\mid\boldX_j)$
is also nonnegative and sums to $\tau_k$ over $\boldy$.
Consequently,
\begin{align}
 \|\hat{P}-\hat{P}^\circ\|_1
 &=\sum_{\boldy\in\{0,1\}^k}
 \bigl(\hat{P}(\boldy)-\hat{P}^\circ(\boldy)\bigr) \notag\\
 &=\sum_{\boldy\in\{0,1\}^k}
 \frac1{J_k}\sum_{j=1}^{J_k}
 \bigl(W_k(\boldy\mid\boldX_j)
       -W_k^\circ(\boldy\mid\boldX_j)\bigr) \notag\\
 &=\frac1{J_k}\sum_{j=1}^{J_k}
 \sum_{\boldy\in\{0,1\}^k}
 \bigl(W_k(\boldy\mid\boldX_j)
       -W_k^\circ(\boldy\mid\boldX_j)\bigr) \notag\\
 &\overset{\eqref{eq:tau_k}}=\frac1{J_k}\sum_{j=1}^{J_k}\tau_k
 =\tau_k.
\label{eq:l1 norm between P and P circ}
\end{align}

We now bound the difference between the full distributions of noisy strings
$\hat{P}$ and $\widetilde{P}$ by writing
$\hat{P}-\widetilde{P}=(\hat{P}-\hat{P}^\circ)
+(\hat{P}^\circ-\widetilde{P}^\circ)+(\widetilde{P}^\circ-\widetilde{P})$
and applying the triangle inequality.
\begin{align}
 \bE\left[\|\hat{P}-\widetilde{P}\|_1\right]
 &\leq
 \bE\left[\|\hat{P}-\hat{P}^\circ\|_1\right]
 +\bE\left[\|\hat{P}^\circ-\widetilde{P}^\circ\|_1\right]
 +\|\widetilde{P}^\circ-\widetilde{P}\|_1 \notag\\
 &\overset{\eqref{eq:l1 norm between mu and mu circ},~\eqref{eq:l1 norm between P and P circ}}=\tau_k+
 \bE\left[\|\hat{P}^\circ-\widetilde{P}^\circ\|_1\right]+\tau_k \notag\\
 &\overset{~\eqref{eq:mixture-l1}}\leq2\tau_k+2^{-k\epsilon/4},
 \label{eq:full-mixture-l1}
\end{align}

Equation~\eqref{eq:full-mixture-l1} bounds the expected
approximation error $\|\hat{P}-\widetilde{P}\|_1$, where the expectation
is over the auxiliary random list
$\boldX_1,\ldots,\boldX_{J_k}$.
Thus, the error averaged over
all sampled lists is at most $2\tau_k+2^{-k\epsilon/4}$, so  at
least one fixed list produces a distribution $Q\in\cV_k$
satisfying
$\|Q-\widetilde{P}\|_1\leq2\tau_k+2^{-k\epsilon/4}$.

Dividing this bound by two provides
\begin{align*}
 \frac12\|\widetilde{P}-Q\|_1
 &\leq \frac12\left(2\tau_k+2^{-k\epsilon/4}\right)=\tau_k+2^{-k\epsilon/4-1}\leq 2e^{-2ks^2}+2^{-k\epsilon/4-1},
\end{align*}
where the last inequality uses
$\tau_k\leq2e^{-2ks^2}$ from
Lemma~\ref{lem:truncated-bsc}.

We also have a second bound that holds for any two probability
distributions.
To verify this directly, nonnegativity provides
$|\widetilde{P}(\boldy)-Q(\boldy)|
\leq\widetilde{P}(\boldy)+Q(\boldy)$ for every possible noisy string $\boldy$.
\begin{align*}
 \frac12\|\widetilde{P}-Q\|_1
 &=\frac12\sum_{\boldy\in\{0,1\}^k}
 |\widetilde{P}(\boldy)-Q(\boldy)|\leq\frac12\left(
 \sum_{\boldy\in\{0,1\}^k}\widetilde{P}(\boldy)
 +\sum_{\boldy\in\{0,1\}^k}Q(\boldy)\right)
 =\frac12(1+1)=1.
\end{align*}
Both upper bounds hold simultaneously, so we may take the
smaller one:
\begin{align*}
 \frac12\|\widetilde{P}-Q\|_1
 &\leq\min\left\{1,\,
 2e^{-2ks^2}+2^{-k\epsilon/4-1}\right\}
 =v_k.
\end{align*}

The choice of $s$ uses only $p$ and $\epsilon$, so $v_k$
does not depend on the distribution $P$. Moreover,
$s>0$ and $\epsilon>0$ remain fixed as $k$ increases.
Thus, both $e^{-2ks^2}$ and $2^{-k\epsilon/4}$ decrease
exponentially, and consequently $v_k\to0$ exponentially.
The same vanishing error bound therefore applies to every
distribution $P$ on the binary strings before the BSC.
\end{proof}

\begin{lemma}
\label{lem:sparse}
For the collections $\cV_k$ in
Construction~\ref{con:sparse-family}, there exists a
nonnegative deterministic sequence $(a_k)_{k\geq1}$,
depending only on $p,\epsilon$, with $a_k\to0$, such that
for every integer $k\geq1$ and every distribution $P$
on $\{0,1\}^k$, some $Q\in\cV_k$ satisfies
$D(\widetilde{P}\|Q)\leq a_k$.
In particular, $A_0=\sup\{a_k\}_{k\geq1}$ is finite and depends only on $p$ and $\epsilon$, and the same $Q$ satisfies $D(\widetilde{P}\|Q)\leq A_0$ for every $k\geq1$ and every distribution $P$ on $\{0,1\}^k$.
\end{lemma}

\begin{proof}
We convert the approximation bound from
Lemma~\ref{lem:fixed-mixture-approximation} into a bound on
$D(\widetilde{P}\|Q)$. Choose the fixed distribution $Q\in\cV_k$
provided by Lemma~\ref{lem:fixed-mixture-approximation}, and let
$d=\frac12\|\widetilde{P}-Q\|_1$. Then $d\leq v_k$, and
\begin{align}\label{eq:l1 norm equals 2d}
 \sum_{\boldy\in\{0,1\}^k}
 |\widetilde{P}(\boldy)-Q(\boldy)|
 &=2d.
\end{align}
This controls the absolute differences between the probabilities
assigned to noisy strings. 

Then, we establish a lower bound on $Q(\boldy)$. Each of
the $k$ factors in $W_k(\boldy\mid\boldx)$ is either $p$
or $1-p$. Since $p<1/2$, both $p$ and~$1-p$ are at least $p$,
so $W_k(\boldy\mid\boldx)\geq p^k$ for every pair of binary
strings. By Construction~\ref{con:sparse-family}, $Q=Q_{\boldx}$ is an
equal-weight linear combination associated with some fixed list
$\boldx=(\boldx_1,\ldots,\boldx_{J_k})$. Hence, for every $\boldy$,
\begin{align*}
 Q(\boldy)
 &=\frac1{J_k}\sum_{j=1}^{J_k}
 W_k(\boldy\mid\boldx_j)
 \geq\frac1{J_k}\sum_{j=1}^{J_k}p^k
 =p^k.
\end{align*}
Also, $Q(\boldy)\leq1$, as $Q$ is a probability
distribution. Therefore,
\begin{align}\label{eq: abstract value of log qy less than klog1p}
 0\leq-\log Q(\boldy)
 &\leq-\log p^k
 =k\log(1/p),
\end{align}
and thus $|\log Q(\boldy)|\leq k\log(1/p)$.

We next bound the difference between the entropies of
$Q$ and $\widetilde{P}$. Applying
Lemma~\ref{lem:entropy-continuity} in Appendix~\ref{app:bernoulli} yields
\begin{align}\label{eq:hq minus hpwk}
 |H(Q)-H(\widetilde{P})|
 &\leq h_2(d)+d\log(2^k)
 =h_2(d)+kd.
\end{align}

Finally, we have
\begin{align}
 D(\widetilde{P}\|Q)
 &=\sum_{\boldy\in\{0,1\}^k}
 \widetilde{P}(\boldy)
 \log\frac{\widetilde{P}(\boldy)}{Q(\boldy)}
 \notag\\
 &=-H(\widetilde{P})
 -\sum_{\boldy\in\{0,1\}^k}
 \widetilde{P}(\boldy)\log Q(\boldy)
 \notag\\
 &=-H(\widetilde{P})
 -\sum_{\boldy\in\{0,1\}^k}
 Q(\boldy)\log Q(\boldy)
 +\sum_{\boldy\in\{0,1\}^k}
 \bigl(Q(\boldy)-\widetilde{P}(\boldy)\bigr)\log Q(\boldy)
 \notag\\
 &=H(Q)-H(\widetilde{P})
 +\sum_{\boldy\in\{0,1\}^k}
 \bigl(Q(\boldy)-\widetilde{P}(\boldy)\bigr)\log Q(\boldy)
 \notag\\
 &\leq |H(Q)-H(\widetilde{P})|
 +\sum_{\boldy\in\{0,1\}^k}
 |Q(\boldy)-\widetilde{P}(\boldy)|\,|\log Q(\boldy)|
 \notag\\
 &\overset{\eqref{eq: abstract value of log qy less than klog1p},~\eqref{eq:hq minus hpwk}}\leq h_2(d)+kd
 +k\log(1/p)\sum_{\boldy\in\{0,1\}^k}
 |Q(\boldy)-\widetilde{P}(\boldy)|
 \notag\\
 &\overset{\eqref{eq:l1 norm equals 2d}}=h_2(d)+kd+2kd\log(1/p).
 \label{eq:reverse-kl}
\end{align}
We now prove that there exists a sequence $(a_k)$,
independent of the distribution $P$, such that
$D(\widetilde{P}\|Q)\leq a_k$ for the selected $Q$
and $a_k\to0$ as $k\to\infty$. 
Since a convergent sequence is bounded, $A_0=\sup_{k\geq1}a_k$ is then finite.

To construct a bound that applies to every $P$, recall
that~\eqref{eq:reverse-kl} bounds $D(\widetilde{P}\|Q)$ in terms
of the distance $d$, and Lemma~\ref{lem:fixed-mixture-approximation} provides
$0\leq d=\frac12\|\widetilde{P}-Q\|_1\leq v_k$. 
We define
\begin{align*}
 a_k
 &=\max_{0\leq t\leq v_k}
 \left\{h_2(t)+kt+2kt\log(1/p)\right\}.
\end{align*}
The maximum exists since $h_2(t)+kt+2kt\log(1/p)$ is continuous in $t$ and $[0,v_k]$ is a closed interval.
Thus, $D(\widetilde{P}\|Q)\leq a_k.$

Now, we turn to prove that $\lim_{k\to\infty}a_k=0$. 
For every $t\in[0,v_k]$, we have
\begin{align*}
 kt\leq kv_k, \qquad
 2kt\log(1/p)\leq2kv_k\log(1/p).
\end{align*}
 Adding these bounds and then taking the maximum over $t\in[0,v_k]$
provides
\begin{align}\label{eq:ak less than three terms}
 a_k
 &\leq kv_k+2kv_k\log(1/p)
 +\max_{0\leq t\leq v_k}h_2(t).
\end{align}
To bound the first two terms $kv_k,2kv_k\log(1/p)$ in~\eqref{eq:ak less than three terms}, recall that \[
0\leq kv_k\overset{\eqref{eq:tv-sequence}}{\leq}2ke^{-2ks^2}+k\,2^{-k\epsilon/4-1}\longrightarrow0\qquad(k\to\infty),
\]
where  $s>0$ and $\epsilon>0$ are fixed as in Lemma~\ref{lem:truncated-bsc} and $\log(1/p)$ is a constant.
Therefore,
$\lim_{k\to\infty}(kv_k+2kv_k\log(1/p))=0.$
To bound the last term $\max_{0\leq t\leq v_k}h_2(t)$ in~\eqref{eq:ak less than three terms},
recall that $v_k\to0$ as $k\to\infty$, so $v_k\leq1/2$
for all sufficiently large $k$. 
Since $h_2$ is increasing on $[0,1/2]$,
$
 \max_{0\leq t\leq v_k} h_2(t)=h_2(v_k)
$
for all sufficiently large $k$. 
Therefore,
\[
 \lim_{k\to\infty}\max_{0\leq t\leq v_k} h_2(t)
 =
 \lim_{k\to\infty}h_2(v_k)
 =
 h_2(0)
 =
 0,
\]
where the convergence follows from the continuity of $h_2$ at zero.
Thus, we have $\lim_{k\to\infty}a_k=0$.
Since $a_k\to0$ as $k\to\infty$, the sequence $(a_k)$
is bounded. 
Thus, let $A_0=\sup_{k\geq1}a_k<\infty$, and
for every $k\geq1$ and every distribution $P$,
 $Q$ satisfies
$D(\widetilde{P}\|Q)\leq a_k\leq A_0$.
\end{proof}

We next compare two joint distributions of $m$ noisy strings
to account for the dependence introduced by random ordering.
Fix $m$ binary strings of length $k$. 
For the first distribution, apply the BSC to each string, and then
arrange the resulting noisy strings in the order given by a permutation
of $\{1,\ldots,m\}$ chosen uniformly at random, independently of the
BSC noise. Each original index is used exactly once, even
when different indices correspond to identical strings.
For the second distribution, choose an original index
uniformly from $\{1,\ldots,m\}$ and apply the BSC to the
string at that index. 
Repeat this selection independently
for each of the $m$ positions, using fresh independent
noise each time. 
An index may therefore be selected several
times or not selected at all. 
We now define these two distributions formally.

\begin{definition}
\label{def:fixed-fragment-output}
Fix positive integers $m,k$ and a tuple
$\boldx=(\boldx_1,\ldots,\boldx_m)$, where
$\boldx_j\in\{0,1\}^k$ for every $j\in[m]$.
From here through the proof of Proposition~\ref{prop:short}, $\boldx$ denotes such a tuple of $m$ strings.
For each $\bolda\in\{0,1\}^k$, define
$m_{\bolda}=|\{j\in[m]:\boldx_j=\bolda\}|$ and
$P_{\boldx}(\bolda)=m_{\bolda}/m$.
The distribution obtained by selecting an index uniformly
and applying the BSC to its string is
\begin{align}
\widetilde{P}_{\boldx}(\boldy)
&=\frac1m\sum_{j=1}^m W_k(\boldy\mid\boldx_j),
\qquad \boldy\in\{0,1\}^k.
\label{eq:empirical-noisy-law}
\end{align}
\begin{itemize}
    \item The first joint distribution, denoted by $\Lambda_{\boldx}$,
is obtained by applying independent BSC noise to the
$m$ fixed strings and then arranging the noisy strings
according to an independent uniform permutation of
$[m]$. 
For every
$(\boldy_1,\ldots,\boldy_m)\in(\{0,1\}^k)^m$,
\begin{align*}
\Lambda_{\boldx}(\boldy_1,\ldots,\boldy_m)
&=\frac1{m!}
\sum_{\pi\text{ a permutation of }[m]}
\prod_{j=1}^m W_k(\boldy_j\mid\boldx_{\pi(j)}).
\end{align*}
Here $\pi(j)$ is the original index assigned to
position $j$. For a fixed permutation, independence
of the BSC noise provides the product of probabilities,
and the factor $1/m!$ averages over all permutations.
Every original index is used exactly once, even when
different indices carry identical strings. The resulting
tuple contains only the noisy strings, without their
original indices.
    \item The second joint distribution, denoted by
$\widetilde{P}_{\boldx}^{\otimes m}$, is obtained by independently
selecting an original index uniformly for each of the
$m$ positions and applying fresh independent BSC noise
to the selected string.
Indices may therefore be selected
repeatedly or omitted. The noisy strings are independent,
each with distribution $\widetilde{P}_{\boldx}$, and
\begin{align*}
\widetilde{P}_{\boldx}^{\otimes m}(\boldy_1,\ldots,\boldy_m)
&=\prod_{j=1}^m\widetilde{P}_{\boldx}(\boldy_j),
\qquad \boldy_1,\ldots,\boldy_m\in\{0,1\}^k.
\end{align*}
\end{itemize}
\end{definition}
The following lemma bounds the relative entropy of the first
joint distribution with respect to the second by $m\log e$.

\begin{lemma}
\label{lem:type}
For the distributions in Definition~\ref{def:fixed-fragment-output},
we have
\begin{equation*}
 D(\Lambda_{\boldx}\|\widetilde{P}_{\boldx}^{\otimes m})\le m\log e.
\end{equation*}
\end{lemma}

\begin{proof}

We now prove the relative-entropy bound. 
It is enough to show that
\begin{align*}
\Lambda_{\boldx}(\boldy_1,\ldots,\boldy_m)
&\leq e^m\widetilde{P}_{\boldx}^{\otimes m}
(\boldy_1,\ldots,\boldy_m)
\end{align*}
for every ordered tuple of noisy strings~$\boldy_1,\ldots,\boldy_m$.
To compare the two joint distributions, we first consider
the binary strings before adding noise. 
The fixed tuple
$\boldx$ contains $m_{\bolda}$ copies of each string
$\bolda\in\{0,1\}^k$.
Permuting its indices preserves
exactly these numbers of copies.
In the second joint distribution, the number of selections of a string may differ from its number of copies in \(\mathbf{x}\).
We condition on selecting exactly
$m_{\bolda}$ copies of every string $\bolda$ among the
$m$ selections. 
We will show that, under this condition,
the selected tuple has the same distribution as a
uniform reordering of $\boldx$. 
The probability of
satisfying this condition will then control the ratio
between the two noisy joint distributions.

Let $\cT_{\boldx}$ be the set of ordered tuples containing
exactly $m_{\bolda}$ copies of each string
$\bolda\in\{0,1\}^k$. 
For any fixed tuple in $\cT_{\boldx}$, the
$m_{\bolda}$ indices carrying $\bolda$ can
be assigned to its positions in $m_{\bolda}!$ ways.
Consequently, every tuple in $\cT_{\boldx}$ is produced
by exactly
$\prod_{\bolda:m_{\bolda}>0}m_{\bolda}!$
index permutations, and
\begin{align*}
|\cT_{\boldx}|
&=\frac{m!}{\prod_{\bolda:m_{\bolda}>0}m_{\bolda}!}.
\end{align*}
Since all index permutations are equally likely, it follows that
uniform reordering provides probability
$1/|\cT_{\boldx}|$ to every tuple in $\cT_{\boldx}$.

We next calculate the probability assigned to this set by independent selection.
For every
$\boldu=(\boldu_1,\ldots,\boldu_m)
\in\cT_{\boldx}$,
\begin{align*}
P_{\boldx}^{\otimes m}(\boldu)
&=\prod_{j=1}^m P_{\boldx}(\boldu_j)
=\prod_{\bolda:m_{\bolda}>0}
\left(\frac{m_{\bolda}}m\right)^{m_{\bolda}}
\end{align*}
and
\begin{align*}
P_{\boldx}^{\otimes m}(\cT_{\boldx})
&=\frac{m!}{\prod_{\bolda:m_{\bolda}>0}m_{\bolda}!}
\prod_{\bolda:m_{\bolda}>0}
\left(\frac{m_{\bolda}}m\right)^{m_{\bolda}}.
\end{align*}
Thus, for every $\boldu\in\cT_{\boldx}$,
\begin{align}
\frac{P_{\boldx}^{\otimes m}(\boldu)}
{P_{\boldx}^{\otimes m}(\cT_{\boldx})}
&=\frac1{|\cT_{\boldx}|}.
\label{eq:type-conditional-uniform}
\end{align}

We now compare the probabilities of a fixed noisy tuple
$(\boldy_1,\ldots,\boldy_m)\in(\{0,1\}^k)^m$.
We first express $\Lambda_{\boldx}$ as an average over uniform
reorderings of $\boldx$ which provides
\begin{align*}
\Lambda_{\boldx}(\boldy_1,\ldots,\boldy_m)
&\overset{}{=}
\frac1{|\cT_{\boldx}|}
\sum_{\boldu\in\cT_{\boldx}}
\prod_{j=1}^m W_k(\boldy_j\mid\boldu_j)
\\
&\overset{\eqref{eq:type-conditional-uniform}}{=}
\frac1{P_{\boldx}^{\otimes m}(\cT_{\boldx})}
\sum_{\boldu\in\cT_{\boldx}}
P_{\boldx}^{\otimes m}(\boldu)
\prod_{j=1}^m W_k(\boldy_j\mid\boldu_j)
\\
&\overset{(\dagger)}{\leq}
\frac1{P_{\boldx}^{\otimes m}(\cT_{\boldx})}
\sum_{\boldu\in(\{0,1\}^k)^m}
\prod_{j=1}^m
P_{\boldx}(\boldu_j)
W_k(\boldy_j\mid\boldu_j)
\\
&\overset{}{=}
\frac1{P_{\boldx}^{\otimes m}(\cT_{\boldx})}
\prod_{j=1}^m
\left(
\sum_{\bolda\in\{0,1\}^k}
P_{\boldx}(\bolda)W_k(\boldy_j\mid\bolda)
\right)
\\
&\overset{(\dagger\dagger)}{=}
\frac{\widetilde{P}_{\boldx}^{\otimes m}
(\boldy_1,\ldots,\boldy_m)}
{P_{\boldx}^{\otimes m}(\cT_{\boldx})}
\end{align*}
where $(\dagger)$ follows by extending the sum from $\cT_{\boldx}$ to all ordered
tuples in $(\{0,1\}^k)^m$ and
$(\dagger\dagger)$ follows from Definition~\ref{def:fixed-fragment-output}:
the sum inside the $j$th factor equals
$\widetilde{P}_{\boldx}(\boldy_j)$.
Multiplying these $m$ factors provides the probability
of the noisy tuple under
$\widetilde{P}_{\boldx}^{\otimes m}$.

Thus, taking logarithms of the probability ratio and averaging under \(\Lambda_{\boldx}\) yields
\begin{align}
D(\Lambda_{\boldx}\|\widetilde{P}_{\boldx}^{\otimes m})
&\leq-\log P_{\boldx}^{\otimes m}(\cT_{\boldx}).
\label{eq:type-probability}
\end{align}
It remains to prove
$P_{\boldx}^{\otimes m}(\cT_{\boldx})\geq e^{-m}$.

Note that,
\begin{align*}
\ln(m!)
&=\sum_{j=1}^m\ln j
\geq\int_1^m\ln t\,dt
=m\ln m-m+1
\geq m\ln m-m.
\end{align*}
Thus, $m!\geq(m/e)^m$.
We have
\begin{align}
P_{\boldx}^{\otimes m}(\cT_{\boldx})
&=\frac{m!}{\prod_{\bolda:m_{\bolda}>0}m_{\bolda}!}
\prod_{\bolda:m_{\bolda}>0}
\left(\frac{m_{\bolda}}m\right)^{m_{\bolda}}
\notag\\
&\overset{(\dagger)}{\geq}
\frac{(m/e)^m}
{\prod_{\bolda:m_{\bolda}>0}m_{\bolda}^{m_{\bolda}}}
\prod_{\bolda:m_{\bolda}>0}
\left(\frac{m_{\bolda}}m\right)^{m_{\bolda}}
\notag\\
&\overset{}{=}
\frac{(m/e)^m}
{\prod_{\bolda:m_{\bolda}>0}m_{\bolda}^{m_{\bolda}}}
\frac{\prod_{\bolda:m_{\bolda}>0}m_{\bolda}^{m_{\bolda}}}
{m^{\sum_{\bolda:m_{\bolda}>0}m_{\bolda}}}
\notag\\
&\overset{(\dagger\dagger)}{=}
\frac{(m/e)^m}{m^m}
=\frac{m^m e^{-m}}{m^m}
=e^{-m}.
\label{eq:type-lower}
\end{align}
Here, $(\dagger)$ uses $m!\geq(m/e)^m$ and
$m_{\bolda}!\leq m_{\bolda}^{m_{\bolda}}$ for every
$m_{\bolda}>0$ and
$(\dagger\dagger)$ follows from $\sum_{\bolda:m_{\bolda}>0}m_{\bolda}=m$.
Therefore,
\begin{align*}
D(\Lambda_{\boldx}\|\widetilde{P}_{\boldx}^{\otimes m})
&\overset{\eqref{eq:type-probability},\,\eqref{eq:type-lower}}{\leq}
-\log(e^{-m})=m\log e. \qedhere
\end{align*}
\end{proof}

In the proof of Proposition~\ref{prop:short} below, we show
\begin{align}
D(\Lambda_{\boldx}\|Q^{\otimes m})
&=D(\Lambda_{\boldx}\|\widetilde{P}_{\boldx}^{\otimes m})
+mD(\widetilde{P}_{\boldx}\|Q)
\label{eq:product-marginal-identity}.
\end{align}
Here $Q\in\cV_k$ is a distribution provided by Lemma~\ref{lem:sparse} for $P=P_{\boldx}$, so that $D(\widetilde{P}_{\boldx}\|Q)\leq A_0$.
The first term measures the dependence among the noisy
strings and is bounded by Lemma~\ref{lem:type}.
The second term adds the approximation error from
Lemma~\ref{lem:sparse} over the $m$ positions.
This provides a bound on $D(\Lambda_{\boldx}\|Q^{\otimes m})$
for each fixed tuple $\boldx$.

For each fixed tuple $\boldx$, we have bounded
$D(\Lambda_{\boldx}\|Q^{\otimes m})$ for a suitable
$Q\in\cV_k$. However, different tuples may require
different choices of $Q$. To turn these bounds into
a mutual-information bound, we introduce one distribution
that is fixed independently of $\boldx$:
\begin{align}
R(\boldy_1,\ldots,\boldy_m)
&=\frac1{|\cV_k|}
\sum_{Q'\in\cV_k}
(Q')^{\otimes m}(\boldy_1,\ldots,\boldy_m).
\label{eq:reference-mixture}
\end{align}
Thus, for each noisy tuple, $R$ averages the probabilities
assigned by all the product distributions in the collection.
Whatever approximation $Q$ is selected for $\boldx$,
its product distribution appears in this average with
weight $1/|\cV_k|$.
All other terms are nonnegative, so
\begin{align*}
R(\boldy_1,\ldots,\boldy_m)
&\geq
\frac{Q^{\otimes m}(\boldy_1,\ldots,\boldy_m)}
{|\cV_k|}.
\end{align*}
Taking logarithms of the probability ratios and averaging
under $\Lambda_{\boldx}$ therefore provides
\begin{align*}
D(\Lambda_{\boldx}\|R)
&\leq D(\Lambda_{\boldx}\|Q^{\otimes m})+\log|\cV_k|.
\end{align*}
This provides a bound relative to the same $R$ for every
tuple $\boldx$.
To explain why this yields an information bound, let
the random tuple $\boldX$ (whose realizations are the fixed tuples~$\boldx$ above) have any joint distribution, let $\boldY$
be its randomly ordered noisy tuple, and denote the
distribution of $\boldY$ by $P_{\boldY}$. Conditionally on $\boldX=\boldx$, the tuple $\boldY$ has distribution $\Lambda_{\boldx}$.
We will use
the identity~\eqref{eq:information-radius}
\begin{align*}
\bE_{\boldX}\!\left[D(\Lambda_{\boldX}\|R)\right]
&=I(\boldX;\boldY)+D(P_{\boldY}\|R)
\geq I(\boldX;\boldY).
\end{align*}
Finally, $\cY$ is obtained
from $\boldY$ by removing its order, so
$I(\boldX;\cY)\leq I(\boldX;\boldY)$.
The resulting bound holds for every joint distribution
of $\boldX$, which proves Proposition~\ref{prop:short}.

\begin{proposition}
\label{prop:short}
For every $\epsilon>0$, there exists a finite constant $A$,
depending only on $p$ and $\epsilon$, such that for all
  integers $m\ge1$ and $k\ge 0$ ,
\begin{equation}
 \Csh(m,k)
 \leq k\left\lceil 2^{k(\kappa+\epsilon)}\right\rceil+Am,
 \label{eq:short-bound}
\end{equation}
where $\Csh$ is defined in Definition~\ref{definition:Csh}
\end{proposition}
\begin{proof}
If $k=0$, then $\Csh(m,0)=0$ and~\eqref{eq:short-bound} holds trivially, so assume $k\geq1$. Fix a tuple $\boldx=(\boldx_1,\ldots,\boldx_m)\in(\{0,1\}^k)^m$, and let $P_{\boldx}$ and $\Lambda_{\boldx}$ be as in Definition~\ref{def:fixed-fragment-output}.
By Lemma~\ref{lem:sparse}, choose $Q\in\cV_k$
satisfying $D(\widetilde{P}_{\boldx}\|Q)\leq A_0$.
We first combine this approximation bound with
Lemma~\ref{lem:type} to bound
$D(\Lambda_{\boldx}\|Q^{\otimes m})$.

For every ordered tuple
$(\boldy_1,\ldots,\boldy_m)\in(\{0,1\}^k)^m$,
where each $\boldy_j$ is one length-$k$ string, we have
\begin{align}
\log\frac{\Lambda_{\boldx}(\boldy_1,\ldots,\boldy_m)}
{Q^{\otimes m}(\boldy_1,\ldots,\boldy_m)}
&=
\log\frac{\Lambda_{\boldx}(\boldy_1,\ldots,\boldy_m)}
{\widetilde{P}_{\boldx}^{\otimes m}(\boldy_1,\ldots,\boldy_m)}
+
\log\frac{\widetilde{P}_{\boldx}^{\otimes m}(\boldy_1,\ldots,\boldy_m)}
{Q^{\otimes m}(\boldy_1,\ldots,\boldy_m)}
\notag\\
&=
\log\frac{\Lambda_{\boldx}(\boldy_1,\ldots,\boldy_m)}
{\widetilde{P}_{\boldx}^{\otimes m}(\boldy_1,\ldots,\boldy_m)}
+
\log\prod_{j=1}^m
\frac{\widetilde{P}_{\boldx}(\boldy_j)}{Q(\boldy_j)}
\notag\\
&=
\log\frac{\Lambda_{\boldx}(\boldy_1,\ldots,\boldy_m)}
{\widetilde{P}_{\boldx}^{\otimes m}(\boldy_1,\ldots,\boldy_m)}
+
\sum_{j=1}^m
\log\frac{\widetilde{P}_{\boldx}(\boldy_j)}{Q(\boldy_j)}.
\label{eq:product-log-split}
\end{align}

We multiply both sides of~\eqref{eq:product-log-split}
by $\Lambda_{\boldx}(\boldy_1,\ldots,\boldy_m)$ and sum over
all $\boldy_1,\ldots,\boldy_m\in\{0,1\}^k$.
We have
\begin{align}
\sum_{\boldy_1,\ldots,\boldy_m\in\{0,1\}^k}
\Lambda_{\boldx}(\boldy_1,\ldots,\boldy_m)
\log\frac{\widetilde{P}_{\boldx}(\boldy_j)}{Q(\boldy_j)}
&\overset{}{=}
\sum_{\boldz\in\{0,1\}^k}
\left(
\sum_{\substack{
\boldy_1,\ldots,\boldy_m\in\{0,1\}^k\\
\boldy_j=\boldz}}
\Lambda_{\boldx}(\boldy_1,\ldots,\boldy_m)
\right)
\log\frac{\widetilde{P}_{\boldx}(\boldz)}{Q(\boldz)}
\notag\\
&\overset{(\dagger)}{=}
\sum_{\boldz\in\{0,1\}^k}
\widetilde{P}_{\boldx}(\boldz)
\log\frac{\widetilde{P}_{\boldx}(\boldz)}{Q(\boldz)}
=D(\widetilde{P}_{\boldx}\|Q),
\label{eq:single-output-divergence}
\end{align}
where $(\dagger)$ follows from the uniform permutation
and~\eqref{eq:empirical-noisy-law}.
Each  index $i\in[m]$ has probability $1/m$
of occupying position $j$. 
Conditional on this choice,
the noisy string at that position equals $\boldz$
with probability $W_k(\boldz\mid\boldx_i)$. Therefore,
\begin{align*}
\sum_{\substack{
\boldy_1,\ldots,\boldy_m\in\{0,1\}^k\\
\boldy_j=\boldz}}
\Lambda_{\boldx}(\boldy_1,\ldots,\boldy_m)
&=\frac1m\sum_{i=1}^m W_k(\boldz\mid\boldx_i)
\overset{\eqref{eq:empirical-noisy-law}}{=}
\widetilde{P}_{\boldx}(\boldz).
\end{align*}
Finally, we have
\begin{align*}
D(\Lambda_{\boldx}\|Q^{\otimes m})
&\overset{\eqref{eq:product-log-split}}{=}
D(\Lambda_{\boldx}\|\widetilde{P}_{\boldx}^{\otimes m})
+
\sum_{j=1}^m
\sum_{\boldy_1,\ldots,\boldy_m\in\{0,1\}^k}
\Lambda_{\boldx}(\boldy_1,\ldots,\boldy_m)
\log\frac{\widetilde{P}_{\boldx}(\boldy_j)}{Q(\boldy_j)}
\\
&\overset{\eqref{eq:single-output-divergence}}{=}
D(\Lambda_{\boldx}\|\widetilde{P}_{\boldx}^{\otimes m})
+mD(\widetilde{P}_{\boldx}\|Q)
\end{align*}
which proves~\eqref{eq:product-marginal-identity}.
Consequently,
\begin{align}
D(\Lambda_{\boldx}\|Q^{\otimes m})
&\overset{\eqref{eq:product-marginal-identity}}{=}
D(\Lambda_{\boldx}\|\widetilde{P}_{\boldx}^{\otimes m})
+mD(\widetilde{P}_{\boldx}\|Q)
\overset{\text{Lemma}~\ref{lem:sparse},~\ref{lem:type}}\leq m\log e+mA_0
=m(\log e+A_0).
\label{eq:product-divergence}
\end{align}

We now compare $\Lambda_{\boldx}$ with the fixed distribution
$R$ defined in~\eqref{eq:reference-mixture}.
The selected $Q$ may depend on $\boldx$, but it always
belongs to $\cV_k$.
Hence, $Q^{\otimes m}$ appears
in the average defining $R$ with weight $1/|\cV_k|$.
Since all other terms are nonnegative, for every
$\boldy_1,\ldots,\boldy_m\in\{0,1\}^k$,
\begin{align}
R(\boldy_1,\ldots,\boldy_m)
&\overset{\eqref{eq:reference-mixture}}{\geq}
\frac1{|\cV_k|}
Q^{\otimes m}(\boldy_1,\ldots,\boldy_m).
\label{eq:reference-component}
\end{align}
Consequently,
\begin{align}
D(\Lambda_{\boldx}\|R)
&=
\sum_{\boldy_1,\ldots,\boldy_m\in\{0,1\}^k}
\Lambda_{\boldx}(\boldy_1,\ldots,\boldy_m)
\log\frac{\Lambda_{\boldx}(\boldy_1,\ldots,\boldy_m)}
{R(\boldy_1,\ldots,\boldy_m)}
\notag\\
&\overset{\eqref{eq:reference-component}}{\leq}
\sum_{\boldy_1,\ldots,\boldy_m\in\{0,1\}^k}
\Lambda_{\boldx}(\boldy_1,\ldots,\boldy_m)
\log\frac{|\cV_k|\Lambda_{\boldx}(\boldy_1,\ldots,\boldy_m)}
{Q^{\otimes m}(\boldy_1,\ldots,\boldy_m)}
\notag\\
&=\log|\cV_k|+D(\Lambda_{\boldx}\|Q^{\otimes m})
\notag\\
&\overset{\eqref{eq:product-divergence}}{\leq}
\log|\cV_k|+m(\log e+A_0)
\notag\\
&\overset{(\dagger)}{\leq}
k\left\lceil2^{k(\kappa+\epsilon)}\right\rceil
+m(\log e+A_0)
\label{eq:reference-divergence}
\end{align}
where $(\dagger)$ follows from
Construction~\ref{con:sparse-family}, which provides
$|\cV_k|\leq2^{kJ_k}$ and
$J_k=\lceil2^{k(\kappa+\epsilon)}\rceil$.

It remains to turn everything into a
mutual-information bound. Allow the random tuple $\boldX$  to have
any joint distribution $P_{\boldX}$.
For its randomly ordered noisy tuple
$\boldY=(\boldY_1,\ldots,\boldY_m)$,  the distribution
satisfies
\begin{align}
P_{\boldY}(\boldy_1,\ldots,\boldy_m)
&=\sum_{\boldx}
P_{\boldX}(\boldx)
\Lambda_{\boldx}(\boldy_1,\ldots,\boldy_m).
\label{eq:noisy-tuple-law}
\end{align}

Then, we have:
\begin{align}
\bE_{\boldX}\!\left[D(\Lambda_{\boldX}\|R)\right]
&=\sum_{\boldx}P_{\boldX}(\boldx)D(\Lambda_{\boldx}\|R)\notag\\
&=
\sum_{\boldx}
\sum_{\boldy_1,\ldots,\boldy_m\in\{0,1\}^k}
P_{\boldX}(\boldx)
\Lambda_{\boldx}(\boldy_1,\ldots,\boldy_m)
\log\frac{\Lambda_{\boldx}(\boldy_1,\ldots,\boldy_m)}
{R(\boldy_1,\ldots,\boldy_m)}
\notag\\
&\overset{}{=}
\sum_{\boldx}
\sum_{\boldy_1,\ldots,\boldy_m\in\{0,1\}^k}
P_{\boldX}(\boldx)
\Lambda_{\boldx}(\boldy_1,\ldots,\boldy_m)
\left[
\log\frac{\Lambda_{\boldx}(\boldy_1,\ldots,\boldy_m)}
{P_{\boldY}(\boldy_1,\ldots,\boldy_m)}
+
\log\frac{P_{\boldY}(\boldy_1,\ldots,\boldy_m)}
{R(\boldy_1,\ldots,\boldy_m)}
\right]
\notag\\
&=
\sum_{\boldx}
\sum_{\boldy_1,\ldots,\boldy_m\in\{0,1\}^k}
P_{\boldX}(\boldx)
\Lambda_{\boldx}(\boldy_1,\ldots,\boldy_m)
\log\frac{\Lambda_{\boldx}(\boldy_1,\ldots,\boldy_m)}
{P_{\boldY}(\boldy_1,\ldots,\boldy_m)}
\notag\\
&\quad+
\sum_{\boldy_1,\ldots,\boldy_m\in\{0,1\}^k}
\left(
\sum_{\boldx}
P_{\boldX}(\boldx)
\Lambda_{\boldx}(\boldy_1,\ldots,\boldy_m)
\right)
\log\frac{P_{\boldY}(\boldy_1,\ldots,\boldy_m)}
{R(\boldy_1,\ldots,\boldy_m)}
\notag\\
&\overset{\eqref{eq:noisy-tuple-law}}{=}
I(\boldX;\boldY)+D(P_{\boldY}\|R)
\geq I(\boldX;\boldY).
\label{eq:information-radius}
\end{align}

Finally, the unordered multiset $\cY$ is obtained
from $\boldY$ by forgetting its order. Thus,
observing $\boldY$ provides at least as much
information about $\boldX$ as observing $\cY$, and
\begin{align*}
I(\boldX;\cY)
&\leq I(\boldX;\boldY)
\overset{\eqref{eq:information-radius}}{\leq}
\bE_{\boldX}\!\left[D(\Lambda_{\boldX}\|R)\right]
\overset{\eqref{eq:reference-divergence}}{\leq}
k\left\lceil2^{k(\kappa+\epsilon)}\right\rceil
+m(\log e+A_0).
\end{align*}
This holds for every joint distribution
$P_{\boldX}$. Taking the supremum proves
\eqref{eq:short-bound} with $A=\log e+A_0$,
which depends only on $p$ and $\epsilon$.
\end{proof}

We next obtain a useful information bound for longer
strings. 
Although Proposition~\ref{prop:short} applies
to every string length~$k$ (in $\Csh(m,k)$, the $m$ strings each have length~$k$), its bound grows exponentially
with the length, making it less useful as~$k$ grows. 
We therefore split each of the $m$
strings of length~$L$ into a prefix of $k$ bits and a
suffix of $L-k$ bits, where $0<k<L$; the boundary cases $k\in\{0,L\}$ are treated in Lemma~\ref{lem:suffix}.
We apply Proposition~\ref{prop:short} to the unordered
noisy prefixes and bound the additional information
provided by the noisy suffixes.
The code may impose relations between prefix and
suffix bits. 
Thus, observing the noisy suffixes may
also help estimate the original prefix bits.
We must bound everything the suffixes reveal after
the noisy prefixes have been observed.
For the proof, we sort the noisy prefixes in lexicographic order, placing equal prefixes in uniform random order.
We arrange the corresponding noisy suffixes in the same order and then consider these two ordered tuples without
revealing the original indices. 
Joining the prefix and suffix at each position recovers every noisy
string. 
 Keeping these
noisy strings, including their multiplicities, without
recording their list positions provides exactly the
multiset observed by the decoder. 
Thus, the two
tuples determine the decoder's entire observation.
They contain at least as much information about
the original strings as that observation, so an
upper bound on their information also applies
to the noisy multiset.

\begin{construction}
\label{con:suffix-sorting}
Fix integers $m\geq1$ and $0<k<L$, and let
$\boldX=(\boldX_1,\ldots,\boldX_m)$ have any joint distribution,
where $\boldX_j\in\{0,1\}^L$ for every $j\in[m]$.
The prefix $\boldX_j[1:k]$ contains the first $k$ bits,
and the suffix $\boldX_j[k+1:L]$ contains the remaining
$L-k$ bits.
Let $\boldZ_j\in\{0,1\}^L$ be the corresponding
noise vector.
All $mL$ noise bits are independent
Bernoulli$(p)$ random variables, and their entire collection
is independent of $\boldX$.
Set $\boldY_j=\boldX_j\oplus\boldZ_j$, where $\oplus$ denotes coordinate-wise
addition modulo two, and let $\cY$ be
the unordered multiset of these $m$ noisy strings. Then
\begin{align}
\boldY_j[1:k]
&=\boldX_j[1:k]\oplus\boldZ_j[1:k],
\qquad
\boldY_j[k+1:L]
=\boldX_j[k+1:L]\oplus\boldZ_j[k+1:L].
\label{eq:y=x+z}
\end{align}
Collect the prefix and suffix noise vectors in their
original index order:
\begin{align*}
\boldZ^{\mathrm{pre}}
&=(\boldZ_1[1:k],\ldots,\boldZ_m[1:k]),
\qquad
\boldZ^{\mathrm{suf}}
=(\boldZ_1[k+1:L],\ldots,\boldZ_m[k+1:L]).
\end{align*}

We then sort the noisy prefixes
$\boldY_j[1:k]$, $j\in[m]$ in lexicographic order, i.e. from smallest to largest
when read as binary numbers. 
If several noisy prefixes are equal, we order their indices~$j$
uniformly at random, independently of $\boldX$ and all BSC noise.
Let $\Pi(i)$ record the original index of the string
placed at position $i$, for $i\in[m]$.
The permutation $\Pi$ is determined by the noisy
prefixes and the independent random choices used
to order equal prefixes. 
Equivalently, it is
determined by $(\boldX,\boldZ^{\mathrm{pre}})$
together with random choices.
Define $\boldA$ as the tuple of noisy prefixes,
$\boldD$ as the tuple of corresponding noisy suffixes,
and $\boldB$ as the tuple of suffixes before noise:
\begin{align}
\boldA
&=(\boldY_{\Pi(1)}[1:k],\ldots,
   \boldY_{\Pi(m)}[1:k]), \notag\\
\boldB
&=(\boldX_{\Pi(1)}[k+1:L],\ldots,
   \boldX_{\Pi(m)}[k+1:L]), \notag\\
\boldD
&=(\boldY_{\Pi(1)}[k+1:L],\ldots,
   \boldY_{\Pi(m)}[k+1:L]).
\label{eq:enhanced-observation}
\end{align}
The decoder forms $\boldA$ and $\boldD$ by
sorting the received noisy strings and separating
each one into its prefix and suffix. 
Joining the entries at position $i$ in $\boldA$
and $\boldD$ provides the complete noisy string
$\boldY_{\Pi(i)}$. Collecting these noisy strings
without their tuple positions provides exactly
$\cY$. Thus, an upper bound on
$I(\boldX;\boldA,\boldD)$ also bounds
$I(\boldX;\cY)$:
\begin{align*}
 I(\boldX;\cY)
 &\leq I(\boldX;\boldA,\boldD).
\end{align*}

\end{construction}

In the setting of Construction~\ref{con:suffix-sorting}, we also collect the suffix noise in the sorted order,
\[
\Pi\boldZ^{\mathrm{suf}}=(\boldZ_{\Pi(1)}[k+1:L],\ldots,\boldZ_{\Pi(m)}[k+1:L]).
\]
Applying~\eqref{eq:y=x+z} at each sorted position shows that $\boldD=\boldB\oplus\Pi\boldZ^{\mathrm{suf}}$. Each $\boldZ_j$ consists of independent Bernoulli$(p)$ bits, but the permutation $\Pi$ depends on the noisy prefixes, and hence on the prefix noise, so it is not immediate that the reordered suffix noise is independent of $(\boldX,\boldA,\boldB)$. The next lemma shows that this independence holds and determines the resulting conditional entropy of~$\boldD$.

\begin{lemma}
\label{lem:sorted-suffix-noise}

Under Construction~\ref{con:suffix-sorting},
\begin{align}
H(\boldD\mid\boldA,\boldB)=H(\boldD\mid\boldX,\boldA,\boldB)=m(L-k)h_2(p).
\label{eq:suffix-noise-entropy}
\end{align}
\end{lemma}

\begin{proof}
We first verify the distribution of the reordered noise
after fixing the quantities used for sorting.
The suffix noise is independent of the collection consisting
of $\boldX$, the prefix noise, and the fresh sorting randomness.
Since $\Pi$ is a function of that collection,
$\boldZ^{\mathrm{suf}}$ is independent of
$(\boldX,\boldZ^{\mathrm{pre}},\Pi)$.
Fix values $\boldX=\boldx$,
$\boldZ^{\mathrm{pre}}=\boldz^{\mathrm{pre}}$, and
$\Pi=\pi$ with positive joint probability. For every
$\boldz\in\{0,1\}^{m(L-k)}$ and with \(\operatorname{wt}(\mathbf{z})\) being the number of ones in \(\mathbf{z}\),
\begin{align}
\Pr\left(
\Pi\boldZ^{\mathrm{suf}}=\boldz
\mid\boldX=\boldx,\,
\boldZ^{\mathrm{pre}}=\boldz^{\mathrm{pre}},\,
\Pi=\pi\right) \notag
&=
\Pr\left(
\boldZ^{\mathrm{suf}}=\pi^{-1}\boldz
\mid\boldX=\boldx,\,
\boldZ^{\mathrm{pre}}=\boldz^{\mathrm{pre}},\,
\Pi=\pi\right) \notag\\
&\quad\overset{(\dagger)}{=}
\Pr\left(\boldZ^{\mathrm{suf}}=\pi^{-1}\boldz\right)
=p^{\wt(\pi^{-1}\boldz)}
(1-p)^{m(L-k)-\wt(\pi^{-1}\boldz)} \notag\\
&\quad\overset{(\dagger\dagger)}{=}
p^{\wt(\boldz)}(1-p)^{m(L-k)-\wt(\boldz)}.
\label{eq:permuted-noise}
\end{align}
The first equality applies the inverse permutation.
Step $(\dagger)$ uses the independence of
$\boldZ^{\mathrm{suf}}$ from
$(\boldX,\boldZ^{\mathrm{pre}},\Pi)$.
Step $(\dagger\dagger)$ uses
$\wt(\pi^{-1}\boldz)=\wt(\boldz)$, since reordering
preserves the number of bits equal to one.
The final expression is the probability of $\boldz$
under independent Bernoulli$(p)$ noise and does not
depend on the conditioning values. Hence $\Pi\boldZ^{\mathrm{suf}}$ consists of independent Bernoulli$(p)$ bits and is independent of $(\boldX,\boldZ^{\mathrm{pre}},\Pi)$.
The tuples $\boldA$ and $\boldB$ in
\eqref{eq:enhanced-observation} are functions of
$(\boldX,\boldZ^{\mathrm{pre}},\Pi)$, so the reordered
noise is also independent of $(\boldX,\boldA,\boldB)$.

To calculate the conditional entropy of $\boldD$,
apply the BSC relation at each sorted position:
\begin{align*}
\boldY_{\Pi(i)}[k+1:L]
&\overset{\eqref{eq:y=x+z}}{=}
\boldX_{\Pi(i)}[k+1:L]
\oplus\boldZ_{\Pi(i)}[k+1:L],
\qquad i\in[m].
\end{align*}
Thus, $\boldD=\boldB\oplus\Pi\boldZ^{\mathrm{suf}}$.
For each fixed value~$\boldb$ of~$\boldB$, addition by~$\boldb$ provides
a one-to-one correspondence between $\boldD$ and
the reordered noise, preserving their conditional
entropy. 
Since the reordered noise is independent of $(\boldX,\boldA,\boldB)$,
it follows that
\begin{align*}
H(\boldD\mid\boldA,\boldB)
&=H(\Pi\boldZ^{\mathrm{suf}}\mid\boldA,\boldB)
=H(\Pi\boldZ^{\mathrm{suf}})\\
&\overset{\eqref{eq:permuted-noise}}{=}
-\bE\left[\log\left(
p^{\wt(\Pi\boldZ^{\mathrm{suf}})}
(1-p)^{m(L-k)-\wt(\Pi\boldZ^{\mathrm{suf}})}
\right)\right]\\
&=-\bE[\wt(\Pi\boldZ^{\mathrm{suf}})]\log p
-\left(m(L-k)-\bE[\wt(\Pi\boldZ^{\mathrm{suf}})]\right)
\log(1-p)\\
&\overset{(\dagger)}{=}
-m(L-k)p\log p-m(L-k)(1-p)\log(1-p)\\
&=m(L-k)h_2(p).
\end{align*}
Here $(\dagger)$ uses the expected number of noise
bits equal to one. Reordering preserves this number,
and each noise bit equals one with probability $p$, so
\begin{align*}
\bE[\wt(\Pi\boldZ^{\mathrm{suf}})]
&=\bE[\wt(\boldZ^{\mathrm{suf}})]
=\sum_{j=1}^{m}\sum_{t=k+1}^{L}p=m(L-k)p.\qedhere
\end{align*} 
\end{proof}

\begin{lemma}
\label{lem:prefix-suffix-information}
Under Construction~\ref{con:suffix-sorting},
\begin{align}
I(\boldX;\boldA)&\leq\Csh(m,k),
\label{eq:sorted-prefix-information}\\
I(\boldX;\boldD\mid\boldA)&\leq m\kappa(L-k).
\label{eq:sorted-suffix-information}
\end{align}
\end{lemma}

\begin{proof}
Conditional on $(\boldX_j[1:k])_{j=1}^m$,
the distribution of $\boldA$ does not depend
on the remaining bits of $\boldX$. Hence,
\begin{align*}
I(\boldX;\boldA)
&=I\bigl((\boldX_j[1:k])_{j=1}^m;\boldA\bigr)
+I\bigl(\boldX;\boldA\mid(\boldX_j[1:k])_{j=1}^m\bigr)
=I\bigl((\boldX_j[1:k])_{j=1}^m;\boldA\bigr)
\overset{(\dagger)}\leq\Csh(m,k),
\end{align*}
where $(\dagger)$ follows from the definition
of $\Csh(m,k)$.

We next bound $I(\boldX;\boldD\mid\boldA)$:
\begin{align*}
I(\boldX;\boldD\mid\boldA)
&\overset{(\dagger)}{=}
I(\boldX,\boldB;\boldD\mid\boldA)
-I(\boldB;\boldD\mid\boldA,\boldX)\leq
I(\boldX,\boldB;\boldD\mid\boldA)\\
&\overset{(\dagger)}{=}
I(\boldB;\boldD\mid\boldA)
+I(\boldX;\boldD\mid\boldA,\boldB)\overset{\text{Lemma}~\ref{lem:sorted-suffix-noise}}{=}
I(\boldB;\boldD\mid\boldA)\\
&=H(\boldD\mid\boldA)-H(\boldD\mid\boldA,\boldB)\overset{\eqref{eq:suffix-noise-entropy}}{=}
H(\boldD\mid\boldA)-m(L-k)h_2(p)\\
&\overset{(\dagger\dagger)}{\leq}
m(L-k)-m(L-k)h_2(p)
=m\kappa(L-k),
\end{align*}
 where $(\dagger)$ follows from the chain rule, the equality marked by Lemma~\ref{lem:sorted-suffix-noise} holds since~\eqref{eq:suffix-noise-entropy} implies $I(\boldX;\boldD\mid\boldA,\boldB)=H(\boldD\mid\boldA,\boldB)-H(\boldD\mid\boldX,\boldA,\boldB)=0$,
and  $(\dagger\dagger)$ follows from
$H(\boldD\mid\boldA)\leq m(L-k)$, since
$\boldD$ contains $m(L-k)$ bits.
\end{proof}

\begin{lemma}
\label{lem:suffix}
For every integer $m\geq1$ and integers $0\leq k\leq L$,
\begin{align}
\Csh(m,L)&\leq\Csh(m,k)+m\kappa(L-k).
\label{eq:suffix-bound}
\end{align}
\end{lemma}

\begin{proof}
We first verify the boundary cases. If $k=L$, the
inequality reduces to $\Csh(m,L)\leq\Csh(m,L)$.
If $k=0$, all prefixes are empty and $\Csh(m,0)=0$.
To prove $\Csh(m,L)\leq mL\kappa$ in this case,
fix any joint distribution of
$\boldX=(\boldX_1,\ldots,\boldX_m)$, with
$\boldX_j\in\{0,1\}^L$ for every $j\in[m]$,
and reveal the original
string indices to the decoder. Let $\boldY$
be the length-$mL$ vector obtained by concatenating
the noisy strings in their original order.
This vector determines $\cY$, so
$I(\boldX;\cY)\leq I(\boldX;\boldY)$.
Moreover, $H(\boldY)\leq mL$, and the independent
BSC errors provide $H(\boldY\mid\boldX)=mLh_2(p)$.
Hence,
\begin{align*}
I(\boldX;\cY)
&\leq I(\boldX;\boldY)
=H(\boldY)-H(\boldY\mid\boldX)
\leq mL-mLh_2(p)=mL\kappa.
\end{align*}
Taking the supremum over all joint distributions
provides $\Csh(m,L)\leq mL\kappa
=\Csh(m,0)+m\kappa(L-0)$.

For $0<k<L$, fix any joint distribution of $\boldX$
and apply Construction~\ref{con:suffix-sorting}.
The pair $(\boldA,\boldD)$ determines $\cY$ by
joining corresponding prefixes and suffixes and
forgetting the order of the resulting noisy strings. Thus,
\begin{align}
I(\boldX;\cY)
&\leq I(\boldX;\boldA,\boldD)
=I(\boldX;\boldA)+I(\boldX;\boldD\mid\boldA) \overset{\eqref{eq:sorted-prefix-information}}{\leq}
\Csh(m,k)+I(\boldX;\boldD\mid\boldA)
\overset{\eqref{eq:sorted-suffix-information}}{\leq}
\Csh(m,k)+m\kappa(L-k).
\label{eq:suffix-information-proof}
\end{align}
Taking the supremum over all joint distributions
of $\boldX$ in~\eqref{eq:suffix-information-proof}
proves~\eqref{eq:suffix-bound}.
\end{proof}

We now combine Proposition~\ref{prop:short} and
Lemma~\ref{lem:suffix} into a bound for
$\Csh(m,L)$ that holds for every fragment count
$m$ and length $L$. 
\begin{proof}[Proof of Proposition~\ref{prop:uniform}]
To choose a prefix length to which the short-multiset bound in Proposition~\ref{prop:short} applies,
set~$k=\min\{L, \lfloor  b\log m\rfloor\}$ and then we have
\begin{equation*}
 \left\lceil2^{k(\kappa+\epsilon)}\right\rceil
 \le2m^{1-\delta},
 \qquad
 k\left\lceil2^{k(\kappa+\epsilon)}\right\rceil
 \le2b\,m^{1-\delta}\log(m+1).
\end{equation*}
The first inequality follows from $k(\kappa+\epsilon)\le b(\kappa+\epsilon)\log m=(1-\delta)\log m$, which shows $1\leq2^{k(\kappa+\epsilon)}\leq m^{1-\delta}$, and from $\lceil x\rceil\leq2x$ for $x\geq1$; the second inequality then follows from $k\leq b\log m\leq b\log(m+1)$. 
Combining
Proposition~\ref{prop:short} with Lemma~\ref{lem:suffix} provides
\begin{align*}
\Csh(m,L)
&\overset{\text{Lemma}~\ref{lem:suffix}}{\leq}
\Csh(m,k)+m\kappa(L-k)\\
&\overset{}{\leq}
k\left\lceil2^{k(\kappa+\epsilon)}\right\rceil
+Am+m\kappa(L-k)\\
&\overset{}{\leq}
2b\,m^{1-\delta}\log(m+1)+Am+m\kappa(L-k)\\
&\overset{(\dagger)}{\mathrel{\leq}}
m(\kappa L-\lambda\log m)^+
+(A+\kappa)m+2b\,m^{1-\delta}\log(m+1).
\end{align*}
where $(\dagger)$ follows from
$L-k\leq(L-b\log m)^++1$ and $\lambda=\kappa b$. 
Indeed, if $L\leq\lfloor b\log m\rfloor$, then $k=L$ and $L-k=0$. Otherwise, $k=\lfloor b\log m\rfloor>b\log m-1$ and $L\geq\lfloor b\log m\rfloor+1>b\log m$, so $L-k<(L-b\log m)+1=(L-b\log m)^++1$. Multiplying by $m\kappa$ shows $m\kappa(L-k)\leq m(\kappa L-\lambda\log m)^++\kappa m$. Hence~\eqref{eq:uniform-shuffling} holds with $c_2=A+\kappa$, which depends only on $p$ and $\epsilon$. 
\end{proof}

We now derive Theorem~\ref{thm:nsc} from Proposition~\ref{prop:uniform}.
\begin{proof}[Proof of Theorem~\ref{thm:nsc}]
By Definition~\ref{def:nsc}, the codeword of the NSC consists of
$K_n=n/\ell_n$ fragments of length~$\ell_n$, so we write its input as
$\boldX=(\boldX_1,\ldots,\boldX_{K_n})$, where
$\boldX_j\in\{0,1\}^{\ell_n}$ is the $j$th fragment. The output $\cY$
is the unordered multiset of the $K_n$ noisy fragments. This is exactly
the setting of Definition~\ref{definition:Csh}, and therefore
$\sup_{P_{\boldX}}I(\boldX;\cY)=\Csh(K_n,\ell_n)$.

Fix $\epsilon>0$ and $\delta\in(0,1)$, and let $b$, $\lambda$, and $c_2$
be as in Proposition~\ref{prop:uniform}; in particular,
$\lambda=(1-h_2(p))(1-\delta)/(1-h_2(p)+\epsilon)$. Applying
Proposition~\ref{prop:uniform} with $m=K_n$ and $L=\ell_n$, and dividing
by $n=K_n\ell_n$, yields
\begin{equation}
\frac{1}{n}\sup_{P_{\boldX}}I(\boldX;\cY)
=\frac{\Csh(K_n,\ell_n)}{n}
\leq\Bigl(1-h_2(p)-\lambda\frac{\log K_n}{\ell_n}\Bigr)^+
+\frac{c_2}{\ell_n}
+\frac{2bK_n^{-\delta}\log(K_n+1)}{\ell_n}.
\label{eq:nsc-bound}
\end{equation}
By Assumption~\ref{remark:assumptions}, $\log n/\ell_n\to\alpha\in(0,\infty)$,
so $\ell_n\to\infty$ and $K_n=n/\ell_n\to\infty$ as $n\to\infty$.
Moreover, $\log K_n/\ell_n=\log n/\ell_n-\log\ell_n/\ell_n\to\alpha$.
Hence, as $n\to\infty$, the first term on the right-hand side
of~\eqref{eq:nsc-bound} tends to $(1-h_2(p)-\lambda\alpha)^+$, and the
second and third terms tend to zero, which shows that
\[
\limsup_{n\to\infty}\frac{1}{n}\sup_{P_{\boldX}}I(\boldX;\cY)
\leq\bigl(1-h_2(p)-\lambda\alpha\bigr)^+.
\]
This bound holds for every $\epsilon>0$ and $\delta\in(0,1)$, and
$\lambda\to1$ as $\epsilon,\delta\to0$, which
proves~\eqref{eq:nsc-converse}.

Fano's inequality, as in the proof of
Corollary~\ref{cor:capacity-limit}, shows that every achievable rate is
at most $(1-h_2(p)-\alpha)^+$. 
Conversely, the index-based scheme
of~\cite{shomorony2019capacity} achieves every rate below
$1-h_2(p)-\alpha$ when this value is positive. Hence the capacity of
the NSC is $(1-h_2(p)-\alpha)^+$.
\end{proof}

\section{Discussion}
\label{sec:scope}
Our converse relies on properties of the binary symmetric channel, and
it would be interesting to extend it to other memoryless channels
acting on the fragments, or to channels with insertions and deletions,
where the noise also changes the lengths of the fragments seen by the
decoder.
The model of~\cite{ravi2024recovering} also allows fragments to be
lost, and a natural next step is to determine the capacity when
fragments can be both lost and corrupted by substitution errors.
In applications such as DNA storage, each fragment may also be read
several times, which suggests extending the converse to noisy fragments
observed with random multiplicities.
More broadly, the two main ingredients of our proof, approximating the
output distributions by a small fixed family and sorting the fragments
by their prefixes, may apply to other channels whose output is an
unordered collection of noisy strings.

\printbibliography
\appendices

\section{Omitted proofs}
\label{app:information}
\label{app:length-entropy}
\label{app:bernoulli}
\label{app:permutation}
\label{app:uniformity}

The next lemma bounds the sums over fragment lengths that appear in the
proof of Lemma~\ref{lem:partition}, when Proposition~\ref{prop:uniform}
is applied to each set of equal-length fragments.
Similar sums appear in~\cite[Appendix~G, eqs.~(89)--(90)]{ravi2024recovering}.
\begin{lemma}
\label{lem:partition-counts}
Fix integers $n\geq2$ and $1\leq K\leq n$, and partition
$n$ into $K$ positive integer lengths. Let $m_L$ be the
number of fragments of length $L$ and
$\cL=\{L\in[n]:m_L>0\}$.
Thus, we have $\sum_{L\in\cL}m_L=K$ and
$\sum_{L\in\cL}Lm_L=n$.
\begin{enumerate}
    \renewcommand{\labelenumi}{(\Alph{enumi})}
    \renewcommand{\theenumi}{(\Alph{enumi})}
\item\label{lemma:A}  We have $|\cL|(|\cL|+1)/2\leq n$, and hence $|\cL|\leq\sqrt{2n}$.
\item\label{lemma:B} 
\begin{align}
\sum_{L\in\cL}m_L\log\frac{n}{m_L}
&\leq 2K\log\frac{n}{K}+K\log e.
\label{eq:counts-entropy}
\end{align}
\item\label{lemma:C} For every fixed $\delta\in(0,1)$,
\begin{align}
\sum_{L\in\cL}m_L^{1-\delta}\log(m_L+1)
&\leq\log(n+1)(2n)^{\delta/2}K^{1-\delta}.
\label{eq:counts-power}
\end{align}
\end{enumerate}
\end{lemma}

\begin{proof}[Proof of Lemma~\ref{lem:partition-counts}. \ref{lemma:A}]
List the distinct fragment lengths in increasing order
as $L_1<\cdots<L_{|\cL|}$. 
Since the lengths are distinct
positive integers, $L_j\geq j$ for every
$j\in\{1,\ldots,|\cL|\}$. 
Moreover, each length occurs
at least once, so $m_L\geq1$ for every $L\in\cL$.
Therefore,
\begin{align*}
n=\sum_{L\in\cL}Lm_L
&\geq\sum_{j=1}^{|\cL|}L_j
\geq\sum_{j=1}^{|\cL|}j
=\frac{|\cL|(|\cL|+1)}{2}. \qedhere
\end{align*}
\end{proof}
\begin{proof}[Proof of Lemma~\ref{lem:partition-counts}. \ref{lemma:B}]
Choose one of the $K$ fragments uniformly and let $V$
be its length. 
Then $\Pr(V=L)=m_L/K$ for every
$L\in\cL$, and
$\bE[V]=\sum_{L\in\cL}Lm_L/K=n/K$.
Using $\sum_{L\in\cL}m_L=K$, we obtain
\begin{align*}
\sum_{L\in\cL}m_L\log\frac{n}{m_L}
&=\sum_{L\in\cL}m_L
\left(\log\frac{n}{K}+\log\frac{K}{m_L}\right)=K\log\frac{n}{K}+KH(V)
\overset{(\dagger)}{\leq}
2K\log\frac{n}{K}+K\log e.
\end{align*}
Here $(\dagger)$ follows from
Lemma~\ref{lem:integer-entropy} in Appendix~\ref{app:bernoulli}, applied to the
positive integer random variable $V$ with mean $n/K$.
\end{proof}

\begin{proof}[Proof of Lemma~\ref{lem:partition-counts}. \ref{lemma:C}]
Fix $\delta\in(0,1)$. Since $m_L\leq n$ for every
$L\in\cL$, we have
\begin{align*}
\sum_{L\in\cL}m_L^{1-\delta}\log(m_L+1)
&\leq\log(n+1)\sum_{L\in\cL}m_L^{1-\delta}
\overset{(\dagger)}{\leq}
\log(n+1)|\cL|^\delta K^{1-\delta}\overset{\text{Lemma}~\ref{lem:partition-counts}. \ref{lemma:A}}{\leq}
\log(n+1)(2n)^{\delta/2}K^{1-\delta}.
\end{align*}
Here $(\dagger)$ follows from the concavity of
$x^{1-\delta}$ and $\sum_{L\in\cL}m_L=K$, which provide
\begin{align*}
\frac1{|\cL|}\sum_{L\in\cL}m_L^{1-\delta}
&\leq
\left(\frac1{|\cL|}\sum_{L\in\cL}m_L\right)^{1-\delta}
=\left(\frac{K}{|\cL|}\right)^{1-\delta}. \qedhere
\end{align*}
\end{proof}
We next state the entropy continuity bound used in the proof of
Lemma~\ref{lem:sparse}.
 \begin{lemma}
\label{lem:entropy-continuity}
If $P$ and $Q$ are distributions on the same alphabet of size $M$
and $d=\frac12\|P-Q\|_1$, then
\begin{equation}
 |H(P)-H(Q)|\le h_2(d)+d\log M.
 \label{eq:entropy-continuity}
\end{equation}
\end{lemma}
\begin{proof}
The claim is trivial when $M=1$. For $M\geq2$, the standard
entropy continuity bound~\cite[Theorem~3]{Sason2013Entropy}
provides, with all logarithms taken to base two,
\[
 |H(P)-H(Q)|
 \leq h_2(d)+d\log(M-1)
 \leq h_2(d)+d\log M.\qedhere
\] 
\end{proof}
The next lemma bounds the entropy of a positive integer random variable
in terms of its mean; it is used in the proof of
Lemma~\ref{lem:partition-counts}. \ref{lemma:B}.
\begin{lemma}
\label{lem:integer-entropy}
Let $V$ be a random variable taking only finitely many
positive integer values and let
$\mu=\bE[V]\in[1,\infty)$ denote its mean.
Then $H(V)\leq\log\mu+\log e$.
\end{lemma}

\begin{proof}
If $\mu=1$, then $V=1$ with probability one, so
$H(V)=0$. 
For $\mu>1$, let $P_V$ denote the distribution of $V$
and introduce the geometric distribution
\begin{align*}
Q(j)&=\frac1{\mu}\left(1-\frac1{\mu}\right)^{j-1},
\qquad j=1,2,\ldots.
\end{align*}
We choose $Q$ to have the same mean $\mu$ as $V$.

Since $V$ takes only finitely many values, the sum
defining $D(P_V\|Q)$ contains finitely many nonzero
terms. Moreover, $Q(j)>0$ for every positive integer
$j$, so each nonzero term is finite. Expanding the
logarithm and using nonnegativity of relative entropy
provides
\begin{align*}
0\leq D(P_V\|Q)
&=\sum_{j:P_V(j)>0}P_V(j)
\log\frac{P_V(j)}{Q(j)}
=-H(V)-\bE[\log Q(V)].
\end{align*}
Rearranging therefore yields
$H(V)\leq-\bE[\log Q(V)]$.

To evaluate this upper bound, the definition of $Q$
provides
$\log Q(V)=-\log\mu+(V-1)\log(1-1/\mu)$.
Since $\bE[V-1]=\mu-1$, we obtain
\begin{align*}
H(V)
&\leq-\bE[\log Q(V)]
=\log\mu-\bE[V-1]\log(1-1/\mu)\\
&=\log\mu-(\mu-1)\log(1-1/\mu)
=\log\mu+(\mu-1)\log\left(1+\frac1{\mu-1}\right)\\
&\overset{(\dagger)}{\leq}\log\mu+\log e.
\end{align*}
Here $(\dagger)$ follows from
$\log(1+x)\leq x\log e$ for $x\geq0$, applied with
$x=1/(\mu-1)$.
\end{proof}
Finally, we state the tail bounds for binomial random variables used in
the proof of Lemma~\ref{lem:truncated-bsc}.
\begin{lemma}
\label{lem:bernoulli}
Let $k$ be a positive integer and let $p\in(0,1)$. Let $Z$ be
Binomial$(k,p)$, that is, the sum of $k$ independent
Bernoulli$(p)$ bits.
For $s>0$ with $p+s<1$,
$\Pr(Z\ge k(p+s))\le e^{-2ks^2}$.
If also $p-s>0$, then
$\Pr(|Z-kp|>ks)\le2e^{-2ks^2}$.
\end{lemma}

\begin{proof}
Write $Z=\sum_{i=1}^{k} X_i$, where the $X_i$ are independent
Bernoulli$(p)$ random variables. Since $0\leq X_i\leq1$ and
$\bE[Z]=kp$, Hoeffding's inequality~\cite{Hoeffding1963} provides,
for every $t>0$,
\[
 \Pr(Z-kp\geq t)\leq e^{-2t^2/k},
 \qquad
 \Pr(Z-kp\leq-t)\leq e^{-2t^2/k}.
\]
Taking $t=ks$ proves the upper-tail bound. Adding the
upper- and lower-tail bounds provides
$\Pr(|Z-kp|>ks)\leq2e^{-2ks^2}$.\qedhere

\end{proof}
\end{document}